\documentclass[twocolumn,superscriptaddress]{revtex4}

\usepackage{mathtools}
\mathtoolsset{showonlyrefs}

\usepackage{bm}
\usepackage{hyperref}

\usepackage{color}
\usepackage{amsfonts}
\usepackage{amsmath}
\usepackage{graphicx}

\usepackage{mathtools}
\usepackage{enumerate}
\usepackage{amssymb}
\usepackage{amsthm}

\newcommand*{\cB}{\mathcal{B}}

\newcommand*{\cF}{\mathcal{F}}

\newcommand*{\cH}{\mathcal{H}}

\newcommand*{\cT}{\mathcal{T}}
\newcommand*{\cV}{\mathcal{V}}

\newtheorem{theorem}{Theorem}
\newtheorem{lemma}[theorem]{Lemma}

\newtheorem{corollary}[theorem]{Corollary}

\newcommand{\ket}[1]{|#1\rangle}

\newcommand*{\Sp}{\mathsf{Sp}}
\newcommand*{\ad}{\mathsf{Ad}}

\newcommand*{\xorig}{x^{\mathsf{or}}}
\newcommand*{\Sorig}{S^{\mathsf{or}}}
\newcommand*{\Xorig}{X^{\mathsf{or}}}
\newcommand*{\Uorig}{U^{\mathsf{or}}}
\newcommand*{\xres}{x^{\mathsf{res}}}
\newcommand*{\Sres}{S^{\mathsf{res}}}

\newcommand*{\Ures}{U^{\mathsf{res}}}
\newcommand*{\xcontrol}{x^{\mathsf{ctr}}}
\newcommand*{\Scontrol}{S^{\mathsf{ctr}}}

\newcommand*{\xtf}{x^{\mathsf{tf}}}
\newcommand*{\Xtf}{X^{\mathsf{tf}}}
\newcommand*{\Stf}{S^{\mathsf{tf}}}
\newcommand*{\Kres}{K^{\mathsf{res}}}

\newcommand*{\yres}{y^{\mathsf{res}}}
\newcommand*{\Mres}{M^{\mathsf{res}}}
\newcommand*{\phires}{\varphi^{\mathsf{res}}}
\newcommand*{\rhores}{\rho^{\mathsf{res}}}
\newcommand*{\zetares}{\zeta^{\mathsf{res}}}

\newcommand*{\Htf}{H^{\mathsf{tf}}}
\newcommand*{\Horig}{H^{\mathsf{or}}}

\newcommand*{\tr}{\mathsf{tr}}

\begin{document}


\title{Universal Uhrig dynamical decoupling for bosonic systems}

\author{Margret Heinze}
\affiliation{Zentrum Mathematik, Technische Universit\"at M\"unchen, 85748 Garching, Germany}
\affiliation{Max-Planck-Institut f\"ur Quantenoptik, 85748 Garching, Germany}
\author{Robert K\"onig}
\affiliation{Zentrum Mathematik, Technische Universit\"at M\"unchen, 85748 Garching, Germany}
\affiliation{Institute for Advanced Study, Technische Universit\"at M\"unchen, 85748 Garching, Germany}

\begin{abstract}
		We construct efficient deterministic dynamical decoupling schemes protecting  continuous variable degrees of freedom. Our schemes target  decoherence induced by quadratic system-bath interactions with analytic time-dependence.
		We show how to suppress such interactions to $N$-th order using only $N$~pulses.  Furthermore, we show to homogenize a $2^m$-mode bosonic system using only $(N+1)^{2m+1}$ pulses, yielding -- up to $N$-th order -- an effective evolution described by non-interacting harmonic oscillators with identical frequencies. The decoupled and homogenized system provides natural decoherence-free subspaces for encoding quantum information. Our schemes only require  pulses  which are tensor products of single-mode passive Gaussian unitaries and SWAP gates between pairs of modes. 
\end{abstract}

\maketitle

Decoherence due to unwanted system-environment interactions is 
a major obstacle on the road towards robust quantum information processing. Although quantum error correction and fault-tolerance provide general  mechanisms to combat such sources of error they are highly demanding in terms of resources. In near-term quantum devices, simpler strategies targeting a reduction of effective error rates at the physical level are more realistic. Dynamical decoupling (DD) is one of the success stories in this direction:  originally developed in the context of NMR~\cite{haeberlenwaugh,waughetal68,hahn50}, it 
has  been  demonstrated in a wide range of systems~\cite{jenista,alvarez10,delange60,du09,wangetal10,bartheletal10,biercuk09}.
In this open-loop control technique, unitary control pulses are instantaneously applied to the system at specific times. The goal is to average out the effect of the system-environment interaction, irrespective of its specific form, approximately 
resulting in a product evolution that acts trivially on the system.

A DD scheme is described by control pulses~$\{U_j\}_{j=1}^L$ applied to the system at times~$\{t_j\}_{j=1}^L\subset [0,T]$ resulting in the evolution
\begin{align}\label{eq:resultingunitary}
	\Ures(T)=U(T,t_L)\prod_{j=1}^L (U_j\otimes I_E) U(t_j,t_{j-1})\ .
\end{align}Here the time evolution $U(t+\Delta t,t)$ from time $t$ to $t+\Delta t$ is generated by the Hamiltonian describing decoherence processes. The pulse sequence achieves $N$-th order decoupling if there is an environment unitary $U_E$ such that $\|\Ures(T)-I_S\otimes U_E\|=O(T^{N+1})$.
For multi-qubit systems, the (nested) Uhrig dynamical decoupling (nested UDD or NUDD) scheme~\cite{mukhtar10,wangliu11,jiang} is the current state of the art: it is remarkably efficient, requiring only~$(N+1)^{2M}$ Pauli pulses to suppress generic interactions between~$M$ qubits and their environment to order~$N$. 
This scaling is significantly more efficient than what could be achieved, e.g., by concatenating~\cite{kholidar} earlier (first order) schemes~\cite{violaknilllloyd99} based on pulses from a unitary~$1$-design applied at equidistant times.
NUDD has been experimentally demonstrated~\cite{singh2017experimental} for three nesting levels.\\

\paragraph*{DD for bosonic systems. }Motivated by the success of DD for qudit systems, one may seek to construct similar protocols for bosonic systems. 
A natural class of system-environment interactions 
are Hamiltonians
\begin{equation}\label{eq:Hamiltonian}
	\Horig(t)=\frac{1}{2}\sum_{j,k=1}^{2n}A_{j,k}(t)R_j R_k
\end{equation} which 
are quadratic in the mode operators
$R=(Q_1^S,\ldots,Q_{n_S}^S,P_1^S,\ldots,Q_{n_S}^S,Q_1^E,\ldots,Q_{n_E}^E,P_1^E,\ldots,P_{n_E}^E)$ of the system and environment; here $A(t)=A(t)^T\in\mathbb{R}^{2n\times2n}$ is symmetric and $n=n_S+n_E$ is the total number of modes~\footnote{In Appendix~\ref{app:linearterms}, we argue that our results extend to the case where the Hamiltonian includes linear terms in the mode operators. }. Such a Hamiltonian generates a one-parameter group of Gaussian unitaries~$U(t)$. Motivated by earlier work giving an example of decoherence suppression in a specific system-environment model~\cite{vitali99}, Arenz, Burgarth and Hillier~\cite{arenz17} have pioneered the systematic study of dynamical decoupling for infinite-dimensional systems. They showed that even in this restricted context, decoupling by application of unitary pulses at specific times cannot be achieved in the same strong sense as for qudit systems: 
While the system's evolution can be (approximately) decoupled from the environment such that $\Ures(T)\approx \Ures_S(T)\otimes \Ures_E(T)$,  no such scheme can render the system's evolution~$\Ures_S(T)$ trivial for an arbitrary initial Hamiltonian~\eqref{eq:Hamiltonian}.
One may, however, find pulse sequences which simplify the system's evolution over time~$T$ to be of the form $\Ures_S(T)\approx e^{i\omega T H_0}$
where $H_0=\frac{1}{2}\sum_{j=1}^{n_S} (Q_j^2+P_j^2)$, a process referred to as {\em homogenization}.  In other words, after applying a homogenization sequence, the effective  decoupled and homogenized evolution 
\begin{align}
\Ures(T)\approx e^{i\omega T H_0}\otimes \Ures_E(T)\label{eq:decoupledhomogenized}
\end{align}
is simply that of identical oscillators (rotating with the same frequency) which do not interact with each other or with the environment.
We remark that an evolution of the form~\eqref{eq:decoupledhomogenized} is still highly beneficial for fault-tolerant quantum information processing as the eigenspaces of the number operator~$H_0$ are now decoherence free~\footnote{The degeneracy of the energy~$E$-eigenspaces of~$H_0$ grows polynomially with~$E$ with a degree determined by the number of modes~$n_S$. }. This provides a way of achieving reduced logical error rates by combining decoupling and homogenization schemes with very simple error-correcting codes spanned by tensor products of number states with fixed total number, such as those constructed in~\cite{ouyang18}.

Here we construct new deterministic schemes that achieve decoupling and homogenization of quadratic system-bath interactions to $N$-th order (for any integer $N$) using only a polynomial number (in $N$) of pulses. 

Our analysis proceeds in the language of the symplectic group~$\mathsf{Sp}(2n)=\{S\in \mathbb{R}^{2n\times2n} | SJ_nS^T=J_n\}$ and its Lie algebra~$\mathfrak{sp}(2n)=\{X\in \mathbb{R}^{2n\times2n} | XJ_n+X^TJ_n=0\}$.
Every Hamiltonian~\eqref{eq:Hamiltonian} can be associated with a symplectic generator $\Xorig(t)=A(t)J_n\in\mathfrak{sp}(2n)$; the former  generates a one-parameter group of Gaussian unitaries~$\Uorig(t)$ which can again be associated with the one-parameter group of symplectic matrices $\Sorig(t)\in\mathsf{Sp}(2n)$ generated by~$\Xorig(t)$.
Instead of~\eqref{eq:resultingunitary} we analyze the resulting symplectic evolution
\begin{align}
	\Sres(T)=S(T,t_L)\prod_{j=1}^L(S_j\oplus I_{2n_E}) \Sorig(t_j,t_{j-1})\  \label{eq:Sres}
\end{align} 
where $\Sorig(t+\Delta t,t)$ is generated by $\Xorig(t)$ from time $t$ to $t+\Delta t$ and the pulses $S_j\in \mathsf{Sp}(2n_S)$ are associated with Gaussian unitary pulses.\\

\paragraph*{A bosonic decoupling scheme.}
We propose the following pulse sequence: the Gaussian unitary $U_S$ defined by its action 
\begin{equation}
U_SQ_iU^*_S= -Q_i\ ,\ U_SP_iU^*_S= -P_i , \textrm{ for } i=1,\ldots,n_S \label{eq:decouplingunitary}
\end{equation}on system mode operators is applied at times
\begin{equation}\label{eq:Uhrigtimes}
t^\textsf{UDD}_j=T\Delta_j\ ,\quad \Delta_j=\sin^2 \frac{j \pi }{2(N+1)}
\end{equation} for $j=1,\ldots,N$. We note that the passive Gaussian unitary~$U_S$ is a tensor product of single-mode phase flips.
\begin{theorem}[Bosonic decoupling sequence] \label{thm:decoupling}
	For any analytic generator $\Xorig:[0,T]\rightarrow\mathfrak{sp}(2(n_S+n_E))$, 
	there are $S_S\in \mathbb{R}^{2n_S\times 2n_S}$ and $S_E\in\mathbb{R}^{2n_E\times 2n_E}$ such that 
	the resulting evolution~\eqref{eq:Sres} after applying~$N$ pulses satisfies
	\begin{equation}\label{eq:symplecticdec_condition}
	\|\Sres(T) -S_S\oplus S_E\|=O(T^{N+1})\ .
	\end{equation}
\end{theorem}

Because of property~\eqref{eq:symplecticdec_condition}, we call the pulse sequence an $N$-th order decoupling scheme. Note that in~\cite{arenz17} a single application of the unitary~$U_S$ was shown to decouple the system from the environment up to first order. To achieve higher order decoupling, remarkably, applying the same pulse~$U_S$ is sufficient and the number of required applications~${N}$ is independent of the number of system and environment modes. 
The times~\eqref{eq:Uhrigtimes} are those associated with the UDD sequence~\cite{uhrig} for a single qubit.  

In Theorem~\ref{thm:decoupling} (and throughout this paper), we state our bounds 
without detailed estimates on the constants in expressions such as $O(T^{N+1})$. For concrete estimates on the required DD control rate~$1/T$, a more refined analysis is necessary. As an example, we provide a corresponding rudimentary bound in Appendix~\ref{app:suffrate}. It involves the
different energy scales~$J_z$ and $J_0$ set by the (uncontrolled) system-environment interaction and the environment Hamiltonian, respectively, and takes the form $O(\frac{1}{(N+1)!}((J_z+J_0)T)^{N+1})$. This mirrors some of the analysis conducted for single qubit UDD in~\cite{uhriglidar}, but we note that the reference also provides more detailed estimates.\\

\paragraph*{Bosonic decoupling with arbitrary pulse times. }
The original derivation of UDD~\cite{uhrig,uhrig2008exact} directly focuses on the effect of $\pi$-pulses (i.e., Pauli-$\sigma_y$) applied at a priori arbitrary times~$t_1,\ldots,t_L$ to a system qubit that is coupled to a bosonic bath. 
The author focuses on a particular figure of merit defined in terms of the overlap of the time-evolved qubit state with the original state. He finds that this ``signal'' is the inverse exponential of a parameter
\begin{align}
\chi (T)= \int_0^\infty \frac{S_\beta(\omega)}{\omega^2} \big|y_L(\omega T) \big|^2 d\omega\label{eq:chiexpressionT}
\end{align}
which depends on the noise spectrum $S_\beta(\omega)$ of the system-bath coupling, as well as the pulse times~$t_1,\ldots,t_L$ via $y_L(z)=1-e^{iz}+ 2\sum_{m=1}^L (-1)^m e^{izt_m/T}$.  Expression~\eqref{eq:chiexpressionT} is then used to find optimal pulse times by   minimizing the quantity~$\chi(T)$. Furthermore, the same expression permits to compare the efficiency of different pulse sequences in a variety of regimes. In particular, it was found that for hard high-frequency cutoffs in $S_\beta(\omega)$, UDD pulse times are optimal, whereas for soft high-frequency cutoffs, the optimal sequences resemble periodic DD~\cite{pasiniuhrigpra10}. 

We argue in Appendix~\ref{app:noisespectrum} that the expression~\eqref{eq:chiexpressionT} also completely characterizes bosonic decoupling for a single mode coupled to a bath of oscillators at inverse temperature~$\beta$: Assuming that the initial state is a product state (with the thermal state of the environment), and the pulse unitary~\eqref{eq:decouplingunitary} is applied at times~$t_1,\ldots,t_L$, we find that the system's resulting evolution is described by a Gaussian quantum channel whose non-unitary component  is fully specified by the quantity~\eqref{eq:chiexpressionT}. Furthermore,~$\chi(T)$ is a direct measure for the degree of non-unitarity. 
This provides a complementary justification for the pulse sequence considered in Theorem~\ref{thm:decoupling}.  Also, all statements about the optimality of pulse sequences and the temperature-dependence of the decoupling efficiency translate immediately from the spin-boson setting to the one considered here.\\

\paragraph*{A bosonic homogenization scheme.}
Assume that the number of modes $n_s=2^m$ is a power of 2 and label the modes by bitstrings~$\nu=(v_1,\ldots,v_m)\in \{0,1\}^m$. Let us introduce the Gaussian unitaries that make up the control pulses~\footnote{We show in Appendix~\ref{app:proofsymplecticparametr} that these are actually passive Gaussian unitaries.}.
	Let $U_{y_0}$ be defined by its action
	\begin{align}\label{eq:uy0def}
	U_{y_0} Q_\nu U_{y_0}^* =P_\nu\quad\textrm{ and }\quad U_{y_0} P_\nu U_{y_0}^* =-Q_\nu
	\end{align}
	for all $\nu\in \{0,1\}^m$, on mode operators. Let us also define for $i=1,\ldots,m$
	the Gaussian unitaries $U_{x_i}$ and $U_{z_i}$ by 
	\begin{align}\label{eq:uxjuzjdef}
	\begin{matrix}
	U_{x_i} Q_{(v_1,\ldots,v_m)}U_{x_i}^*& =&Q_{(v_1,\ldots,v_{i-1},1-v_i,v_{i+1},\ldots,v_m)}\\
	U_{x_i} P_{(v_1,\ldots,v_m)}U_{x_i}^* &=&P_{(v_1,\ldots,v_{i-1},1-v_i,v_{i+1},\ldots,v_m)}\\
	U_{z_i} Q_{(v_1,\ldots,v_m)}U_{z_i}^*& =&(-1)^{v_i}Q_{(v_1,\ldots,v_m)}\\
	U_{z_i} P_{(v_1,\ldots,v_m)}U_{z_i}^* &=&(-1)^{v_i}P_{(v_1,\ldots,v_m)}
	\end{matrix}
	\end{align}
	for all $\nu =(v_1,\ldots,v_m)\in \{0,1\}^m$. We set $U_{y_i}=U_{x_i}U_{z_i}$. 
	The unitary~$U_{y_0}$ acts as a tensor product of the same passive single-mode Gaussian unitary on all modes, $U_{z_i}$ as the tensor product of single-mode phase flips on half of the modes and~$U_{x_i}$ is a product of two-mode SWAP gates between pairs of modes. Depending on the experimental setup, the difficulty of realizing 
	two-mode SWAP gates may differ significantly from that associated with single-mode passive gates. Unlike in the case of multi-qubit DD schemes (which only require single-qubit Pauli gates), this fact needs to be taken into account when analyzing e.g., the effect of finite pulse widths.

Using these unitaries, we show how to construct a multi-mode homogenization scheme from a multi-qubit DD scheme. More precisely, assume that an $(m+1)$-qubit DD scheme with qubits labeled from 0 to $m$ uses pulses which are products of single-qubit Pauli matrices $(\sigma_w)_k$ where $w\in \{x,y,z\} $ and $k=0,\ldots,m$. Then we construct the bosonic pulses by replacing Pauli  factors (retaining their order in the product) according to the substitution rules 
\begin{align}
	\begin{matrix}
	(\sigma_x)_0 & \mapsto &U_{y_0},\\
	(\sigma_y)_0 & \mapsto &U_{y_0},\\
	(\sigma_z)_0 & \mapsto &I,
	\end{matrix}\qquad \quad
	\begin{matrix}
	(\sigma_x)_i & \mapsto &U_{x_i},\\
	(\sigma_y)_i & \mapsto &U_{y_i},\\
	(\sigma_z)_i & \mapsto &U_{z_i},
	\end{matrix}\ \label{eq:replacements}
\end{align} where $k=1,\ldots,m$ and $I$ means that no pulse is applied. 
Our homogenization scheme is obtained by applying the substitution rule~\eqref{eq:replacements} to the NUDD scheme~\cite{mukhtar10,wangliu11,jiang} for $m+1$ qubits. This results in the following:

\begin{figure*}[t]
	\includegraphics[width=\textwidth]{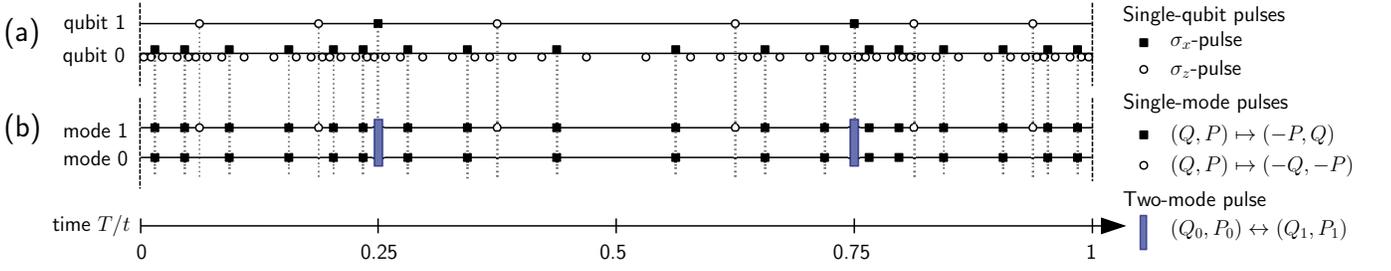}
	\caption{Bosonic homogenization scheme~(b) of suppression order~$N=2$ for two modes and the corresponding order-2 NUDD scheme~(a).\label{fig:bosonic_hom_2modesorder2} 
		The evolution describing decoherence (horizontal straight lines) is interleaved with instantaneous control pulses.}
\end{figure*}

\begin{theorem}[Bosonic homogenization sequence]\label{thm:homogenization}
	The described pulse sequence consists of $(N+1)^{2m+1}$ passive Gaussian pulses. For any analytic generator $\Xorig:[0,T]\rightarrow\mathfrak{sp}\big(2 (2^m+n_E)\big)$ of the form $\Xorig_S(t)\oplus \Xorig_E(t)$, there are $\omega\in \mathbb{R}$ and $S_E\in\Sp(2n_E)$ such that the resulting evolution~\eqref{eq:Sres} satisfies
	\begin{equation}\label{eq:homogenizationcondition}
		\|\Sres(T)-e^{\omega T J_{2^m}}\oplus S_E\|=O(T^{N+1})\ .
	\end{equation} 
\end{theorem}

Theorem~\ref{thm:homogenization} assumes that system and environment are already decoupled, i.e., the original Hamiltonian~\eqref{eq:Hamiltonian} is of the form $\Horig(t)=H_S(t)\otimes I_E+I_S\otimes H_E(t) $; it guarantees that the system's evolution is homogenized since the symplectic generator $J_{2^m}\in \mathfrak{sp}(2\cdot2^m)$ is associated with $H_0$. Correspondingly, we call the pulse sequence constructed here a universal bosonic homogenization sequence of order~$N$. Combining decoupling and homogenization schemes (by concatenation~\footnote{We remark that Theorem~\ref{thm:homogenization} extends to (not necessarily symplectic) generators $\Xorig(t):[0,T]\rightarrow\mathbb{R}^{2n\times2n}$. Therefore, it can be applied in conjunction with Theorem~\ref{thm:decoupling}.}) leads to an effective evolution of the form~\eqref{eq:decoupledhomogenized}. In the remainder of this paper, we sketch the proofs of Theorems~\ref{thm:decoupling} and~\ref{thm:homogenization}. \\

\paragraph*{Bosonic decoupling using Uhrig times.}
To prove Theorem~\ref{thm:decoupling}, we use the direct-sum structure of the matrix~$\Xorig(t)\in\mathfrak{sp}(2(n_S+n_E))$, that is
\begin{equation}\label{eq:Xbosdec}
\Xorig(t)=\begin{pmatrix}
X_{SS}(t) &X_{SE}(t)\\
X_{ES}(t)& X_{EE}(t)
\end{pmatrix}\ .
\end{equation} Here $X_{BC}(t)$ are analytic functions of real $2n_B\times 2n_C$ matrices for $B,C\in \{S,E\}$ by assumption;
$X_{SE}(t)$  and $X_{ES}(t)$ are responsible for system-environment interactions. 
We define the piecewise constant function $\sigma:[0,1]\rightarrow \{-1,1\}$ that satisfies $\sigma(0)=1$ and switches its sign whenever the pulse~$-I_S\in \Sp(2n_S)$ (associated with the unitary $U_S$) is applied, i.e., at each~$\{t_j/T\}_{j=1}^{N}$. For our analysis, we change into the toggling frame~\footnote{Let $\Sorig(t),\Scontrol(t)\in \Sp(2n)$ denote the original and the control evolution, respectively. Then the toggling frame evolution is defined as $\Stf(t)={\Scontrol(t)}^{-1} \Sorig(t)$ with generator $\Xtf(t)={\Scontrol(t)}^{-1}\Xorig(t) \Scontrol(t)\in \mathfrak{sp}(2n)$. The pulse sequences considered throughout this paper consist in applying instantaneous pulses~$S_j\in \Sp(2n_s)$, such that~$\Scontrol(t)$ is simply the product of all pulses applied up to time $t$.} with evolution~$\Stf(T)$ generated by
\begin{equation}
\Xtf(t)=\begin{pmatrix}
X_{SS}(t)& \sigma(t/T)X_{SE}(t)\\
\sigma(t/T)X_{ES}(t) & X_{EE}(t)
\end{pmatrix}\ . \label{eq:xtildegeneralop}
\end{equation}
Direct computation of the Dyson series of~$\Stf(T)$ in Appendix~\ref{app:proofdec} shows that a sufficient condition for $N$-th order decoupling is the following:

The function~$\sigma(t)$ is (or equivalently the times $t_j$ are) a solution to the integral equations
\begin{equation}
\cF_{\gamma_1,\ldots,\gamma_s}^{r_1,\ldots,r_s}(\sigma)=0\quad
\textrm{if}\quad
\begin{cases}
s+\sum_{k=1}^s r_k\leq N\textrm{ and}\\
\bigoplus_{k=1}^s \gamma_k=1\ 
\end{cases}	\label{eq:conditionfsigma}
\end{equation}
for all $s\in\mathbb{N}$, $r_1,\ldots,r_s\in\mathbb{N}_0$ and $\gamma_1,\ldots,\gamma_s\in \mathbb{Z}_2$, where 
$\cF_{\gamma_1,\ldots,\gamma_s}^{r_1,\ldots,r_s}(\sigma)=\int_0^1 d\tau_s\cdots \int_0^{\tau_{2}} d\tau_1 \prod_{k=1}^s \sigma(\tau_k)^{\gamma_k}\tau_k^{r_k}$  and where $\oplus$ denotes addition modulo 2.

The same integral equations~\eqref{eq:conditionfsigma} appear in the analysis of the UDD scheme for a single qubit~\cite{yang2008universality,jiang}; in particular, it is known that the times~\eqref{eq:Uhrigtimes} are a solution. Thus we obtain a bosonic decoupling scheme (for an arbitrary number of modes) by using the times of the single-qubit UDD sequence. \\

\paragraph*{Multiple qubits and multiple bosonic modes.}
The proof of Theorem~\ref{thm:homogenization} relies on a connection between multi-qubit systems and bosonic systems: we identify elements of~$\Sp(2\cdot 2^m)$ and a basis of its Lie algebra~$\mathfrak{sp}(2\cdot2^m)$ which satisfy commutation relations analogous to those obeyed by the Pauli matrices. We associate mode operators with basis vectors of $\mathbb{R}^{2\cdot 2^m}\cong  \mathbb{R}^2\otimes (\mathbb{R}^2)^{\otimes m}$
by
\begin{align}
\begin{matrix}
	Q_{(v_1,\ldots,v_m)} &\quad \leftrightarrow\qquad \ket{q}\otimes \ket{e_{v_1}}\otimes\cdots\otimes \ket{e_{v_m}}\\
	P_{(v_1,\ldots,v_m)} &\quad \leftrightarrow\qquad \ket{p}\otimes \ket{e_{v_1}}\otimes\cdots\otimes \ket{e_{v_m}}
\end{matrix}\ . \label{eq:tensorsymplectic}
\end{align}  Here we use an orthonormal basis $\ket{q},\ket{p}$ of~$\mathbb{R}^2$ for the first factor (which we will later identify with `qubit 0'), as well as an orthonormal basis $\ket{e_0},\ket{e_1}$ for each of the remaining $m$ factors (which will be identified with `qubits 1 to $m$').
On $\mathbb{R}^2$, let us define the matrices
\begin{equation}
I=\begin{pmatrix}
1&0\\0&1
\end{pmatrix},\;
x=\begin{pmatrix}
0&1\\1&0
\end{pmatrix},\;
y=\begin{pmatrix}
0 & -1\\ 1&0
\end{pmatrix},\;
z=\begin{pmatrix}
1 & 0\\ 0&-1
\end{pmatrix}
\end{equation} that we also write as $S_{(0,0)}$, $S_{(1,0)}$, $S_{(1,1)}$, $S_{(0,1)}$, respectively. 

\begin{lemma}\label{lem:parametrizesymplecticgroupalgebra}
	For $\alpha=(a_0,a_1,\ldots,a_m)\in(\mathbb{Z}_2^2)^{m+1}$, define the matrix
$S_{\alpha}=S_{a_0}\otimes S_{a_1}\otimes\cdots\otimes S_{a_m}$ on $\mathbb{R}^2\otimes (\mathbb{R}^2)^{\otimes m}$. 
	\begin{enumerate}[(a)]
			\item There is a subset $\Gamma\subset (\mathbb{Z}_2^2)^{m+1}$ such that $\{S_\alpha\}_{\alpha\in\Gamma}$ is a basis of the Lie algebra~$\mathfrak{sp}(2\cdot2^m)$. \label{it:symplecticalgebra}
			\item Let $\tilde{\Gamma}$ be the set of sequences $\alpha=(a_0,a_1,\ldots,a_m)\in (\mathbb{Z}_2^2)^{m+1}$
			such that $a_0\in \{(0,0),(1,1)\}$. Then~$S_\beta$ is orthogonal symplectic for every $\beta\in\tilde{\Gamma}$. \label{it:symplecticgroup}
			\item The symplectic form is given by  $J_{2^m}=-S_{(1,1,0,\ldots,0)}$.\label{it:symplecticformasssmatrix}
			\item The adjoint action of $S_\beta\in \Sp(2\cdot 2^m)$ is 
			\begin{equation}\label{eq:adjointactionsymplectic}
				S_\beta^{-1} S_\alpha S_\beta=(-1)^{\langle\alpha,\beta\rangle}S_\alpha 
			\end{equation}
		for all $ \alpha\in{\Gamma}$, $\beta\in\tilde{\Gamma}$, where
		\begin{equation}
			\langle \alpha,\beta\rangle =\sum_{j=0}^{m}  a_j^T \begin{pmatrix}
			0&1\\-1&0
			\end{pmatrix}b_j
		\end{equation} 
		is the usual symplectic inner product on $(\mathbb{Z}_2^2)^{m+1}$. \label{it:adjointaction}
	\end{enumerate}  
\end{lemma}
\noindent The proof of this Lemma is given in Appendix~\ref{app:proofsymplecticparametr}. There we also show that the homogenization unitaries $U_{x_i},U_{y_j},U_{z_i}$ from Eqs.~\eqref{eq:uy0def} and~\eqref{eq:uxjuzjdef} are associated with the symplectic matrices $x_i,y_j,z_i\in \{S_\beta\ |\ \beta\in\tilde{\Gamma}\}$ for $i\ge1$, and $j\ge0$. 
Note that the relations~\eqref{eq:adjointactionsymplectic} are formally analogous to the commutation relations $\sigma_\beta^{-1} \sigma_\alpha \sigma_\beta = (-1)^{\langle\alpha,\beta\rangle}\sigma_\alpha$ of Pauli operators for $m+1$ qubits~\footnote{Here, for $\alpha=(a_0,a_1,\ldots,a_m)\in (\mathbb{Z}_2^2)^{m+1}$, the multi-qubit Pauli operators are defined as $\sigma_\alpha=\sigma_{a_0}\otimes \sigma_{a_1}\otimes\cdots\otimes \sigma_{a_m}$ where the factors in the tensor product are the single-qubit Pauli matrices $\sigma_{(0,0)}=I$, $\sigma_{(1,0)}=\sigma_x$, $\sigma_{(1,1)}=\sigma_y$, $\sigma_{(0,1)}=\sigma_z$. }.\\

\paragraph*{Bosonic homogenization from qubit DD.}
The close resemblance of the commutation relations~\eqref{eq:adjointactionsymplectic} with those of Pauli matrices is key to our construction of homogenization schemes.  We remark that for qubit DD schemes with Pauli pulses, it is precisely the phases~$(-1)^{\langle \alpha, \beta\rangle}$ that lead to a cancellation of unwanted terms in the effective evolution. However, in contrast to the qubit setting, the set of available pulses  in the bosonic setting is restricted: only matrices in~$\{ S_\beta|\ \beta \in\tilde{\Gamma}\}$, i.e.,~that act as~$y$ or~$I$ on  `qubit 0' (the first factor of~\eqref{eq:tensorsymplectic}) are available as control pulses. This motivates the substitution rules~\eqref{eq:replacements}, where we replace Pauli operators~$(\sigma_w)_k$ by symplectic matrices~$w_k$ for $k\in \{1,\ldots,m\}$ while on `qubit 0' we only allow~$y$ and~$I$.

To analyze the effect of the resulting pulse sequence on a decoupled evolution, suppose the original generator $\Xorig(t)=\Xorig_S(t)\oplus \Xorig_E(t)\in\mathfrak{sp}\big(2\cdot (2^m+n_E)\big)$ satisfies
\begin{equation}
\Xorig_S(t)=\sum_{\alpha\in \Gamma} B_\alpha(t)S_\alpha \ , \textrm{ where } B_\alpha(t)=\sum_{r=0}^\infty b_{\alpha,r}t^r
\end{equation} 
for $b_{\alpha,r}\in \mathbb{R}$ and where we use the basis of~$\mathfrak{sp}(2\cdot 2^m)$ from Lemma~\ref{lem:parametrizesymplecticgroupalgebra}.
Suppose the control sequence is defined by the times~$\{t_j\}_{j=1}^L$ and a function~$\beta:\{1,\ldots,L\}\to \tilde{\Gamma}$ specifying which control pulse~$S_{\beta(j)}$ is applied at time $t_j$. 
Since the original evolution is decoupled, it is sufficient to restrict to the system only (and omit the index $S$).
It is convenient to change into the toggling frame with evolution $\Stf(T)$ after time $T$.
By exploiting the parametrization of the symplectic Lie group and its Lie algebra introduced in Lemma~\ref{lem:parametrizesymplecticgroupalgebra} and using the relations~\eqref{eq:adjointactionsymplectic}, we conclude that the toggling frame generator takes the form
\begin{equation}\label{eq:togglinggenerator}
	\Xtf(t)=\sum_{\alpha\in\Gamma} F_\alpha(t/T)B_\alpha(t) S_\alpha 
\end{equation} where we defined the functions
\begin{equation}\label{eq:Falphahomogenization}
	F_\alpha(t/T)=(-1)^{\sum_{j: t_j\le t} \langle \alpha, \beta(j)\rangle} \quad \textrm{ for } t\in [0,T]\ .
\end{equation}
Using the generator's form~\eqref{eq:togglinggenerator}, the toggling frame evolution
$\Stf(t)=\cT\exp\left[\int_0^t \Xtf(\tau)d\tau\right]$ can be expanded in a Dyson series as
\begin{equation}\label{eq:togglingframeDyson}
	\Stf(t)=\sum_{s=0}^\infty \sum_{\vec{\alpha}\in \Gamma^s}\sum_{{\vec{r}}=0}^\infty \prod_{k=1}^s S_{\alpha_k}b_{\alpha_k,r_k}\cF_{\vec{\alpha}}^{\vec{r}} \ T^{s+\sum_{k=1}^sr_k}
\end{equation} where $\vec{\alpha}=(\alpha_1,\ldots,\alpha_s)$, ${\vec{r}}=(r_1,\ldots,r_s)$ and where
\begin{equation}\label{eq:togglingintegrals}
	\cF_{\vec{\alpha}}^{\vec{r}}(\{F_\alpha\}) =\int_0^1d\tau_s\cdots\int_0^{\tau_{2}}d\tau_1 \prod_{k=1}^s F_{\alpha_k}(\tau_k)\tau_k^{r_k}\ .
\end{equation} We can directly read off the $N$-th order term from~\eqref{eq:togglingframeDyson}.
Furthermore, it is easy to see that an approximate system's evolution of the form $c_1 I_S-c_2J_{2^m}$ for some $c_1,c_2\in\mathbb{R}$ is homogenized (cf.~Appendix~\ref{app:proofhomogen} for the calculation). Hence, the following condition is sufficient to achieve homogenization up to order~$N$: 
for  $s\in \mathbb{N}$, $r_1,\ldots,r_s\in \mathbb{N}_0$, and $\alpha_1, \ldots, \alpha_s\in \Gamma$, we have
\begin{align}
\cF_{\vec{\alpha}}^{\vec{r}}(\{F_\alpha\}) =0\quad \textrm{ if }
\begin{cases}
s+\sum_{k=1}^s r_k\le N\textrm{ and }\\
\prod_{k=1}^s S_{\alpha_k} \not\in \{\pm I_S,\pm  J_{2^m}\}\ .
\end{cases}	\label{eq:conditionshom1}
\end{align}
A similar analysis applies to $(m+1)$-qubit DD schemes with multi-qubit Pauli pulses, see~\cite{jiang}: 
Identically defined functions~\eqref{eq:Falphahomogenization} appear in the toggling frame generator and give rise to the same coefficients~\eqref{eq:togglingintegrals} in the Dyson series of the toggling frame evolution.
As a consequence, any qubit decoupling scheme based on Pauli pulses provides the necessary cancellations when translated to the bosonic homogenization setting using the substitution rule~\eqref{eq:replacements}.

In more detail, let a universal $N$-th order $(m+1)$-qubit DD scheme be defined by $L\in\mathbb{N}$ pulses $U_j=\sigma_{\beta(j)}$ for $\beta: \{1,\ldots,L\}\to (\mathbb{Z}^2_2)^{m+1}$ that are applied to the system at times $t_j$. 
Then the functions~$F^\textsf{qubit}_\alpha$ defined by~\eqref{eq:Falphahomogenization} for~$\alpha\in (\mathbb{Z}^2_2)^{m+1}$ satisfy
\begin{align}
\cF_{\vec{\alpha}}^{\vec{r}}(\{F^\textsf{qubit}_\alpha\}) =0\quad\textrm{ if }
\begin{cases}
s+\sum_{k=1}^s r_k\le N\textrm{ and }\\
\bigoplus_{k=1}^s \alpha_k \neq (0,\ldots,0)\ .
\end{cases}	\label{eq:conditionsqubitdec}
\end{align}
for all $s\in \mathbb{N}$, $r_1,\ldots,r_s\in \mathbb{N}_0$, and $\alpha_1, \ldots, \alpha_s\in (\mathbb{Z}^2_2)^{m+1}$. 

The associated bosonic homogenization scheme 
(obtained using the substitution rule~\eqref{eq:replacements}) then has toggling frame generator specified by functions $F^\textsf{bos}_\alpha$ defined as
\begin{equation}\label{eq:homFalphabeta'}
        F^\textsf{bos}_\alpha(t/T)=(-1)^{\sum_{j: t_j\le t} \langle \alpha, \beta'(j)\rangle} \quad \textrm{ for } t\in [0,T]\ .
\end{equation} for all $\alpha\in \Gamma$. Here $\beta'(j)\in(\mathbb{Z}^2_2)^{m+1}$  differs from $\beta(j)\in(\mathbb{Z}^2_2)^{m+1}$
    only in the first entry (associated with qubit~$0$), where $(1,0)$ (respectively~$(0,1)$) is replaced by $(1,1)$ (respectively $(0,0)$) as prescribed by~\eqref{eq:replacements}. 
    With~\eqref{it:symplecticformasssmatrix}, it is straightforward to verify (see Appendix~\ref{app:proofhomogen} for the details) that property~\eqref{eq:conditionsqubitdec}
    of the functions $F^{\textsf{qubit}}_\alpha$ implies the desired property~\eqref{eq:conditionshom1} for the functions~$F^{\textsf{bos}}_\alpha$. In other words, the decoupling property in the qubit setting translates to homogenization of bosonic modes. 

Having established a general connection between universal $(m+1)$-qubit DD schemes and bosonic homogenization of~$2^m$ modes, Theorem~\ref{thm:homogenization} follows immediately by applying this to the NUDD sequence. The latter achieves $N$-th order decoupling of $(m+1)$~qubits with $(N+1)^{2(m+1)}$ Pauli pulses and is defined recursively by concatenation of Uhrig sequences (cf.~Appendix~\ref{app:udd} for a revision of the NUDD sequence). Examples of the resulting bosonic homogenization schemes are shown in Fig.~\ref{fig:bosonic_hom_2modesorder2} and Fig.~\ref{fig:bosonic_hom_4modesorder1}.\\

\paragraph*{Conclusions. }
Our work introduces novel, highly efficient dynamical decoupling schemes for bosonic systems. Instead of  applying finite-dimensional (qubit) decoupling procedures to distinguished subspaces, our schemes are of a genuinely continuous-variable nature. This leads to remarkably simple schemes involving only passive Gaussian unitaries. On a conceptual level, our work establishes a tight connection between qubit- and continuous-variable schemes. In particular, it implies for example that considerations related to pulse imperfections such as finite widths (see e.g.,~\cite{yang2008universality,uhrig2010efficient}) translate immediately to our bosonic schemes. More generally, this analogy may be useful elsewhere to lift qubit information processing primitives to the bosonic context. On a practical level, we believe that our protocols could become a powerful tool for continuous-variable quantum information processing as they pose minimal experimental requirements.\\

\paragraph*{Acknowlegdements. }
RK is supported by the Technische Universit\"at M\"unchen -- Institute for Advanced Study funded by the German Excellence Initiative and the European Union Seventh Framework Programme under grant agreement no.~291763. He acknowledges support by the German Federal Ministry of Education through the funding program Photonics Research Germany, contract no.~13N14776 (QCDA-QuantERA). MH is supported by the International Max Planck Research School for Quantum Science and Technology at the Max-Planck-Institut f\"ur Quantenoptik.

\bibliographystyle{h-physrev}
\bibliography{literature}

\clearpage
\newpage
\appendix

\section{General decoupling condition}\label{app:generaldecoupling}  
Here we revisit the analysis of decoupling schemes using the toggling frame and Dyson expansion. We formulate this in terms of general matrix Lie groups. This will be convenient since both multi-qubit decoupling as well as bosonic homogenization schemes fall into this setup. In Section~\ref{app:udd} we give a brief summary of the UDD and NUDD schemes. The former will be related to our bosonic decoupling scheme and the latter to the homogenization scheme.

\subsection{Setup}\label{app:setup}
The setup is as follows. Let $G\subset \mathsf{GL}(\mathbb{C}^{d_S})$ or $G\subset \mathsf{GL}(\mathbb{R}^{d_S})$ be a matrix Lie group associated with a ``system'' $\cH_S=\mathbb{C}^{d_S}$ or $\cH_S=\mathbb{R}^{d_S}$, respectively. For $m$ qubits we will identify $G=\mathsf{U}(2^{m})$ whereas~$G=\Sp(2n_S)$ in the bosonic setting with $n_S$ modes. We will assume that the Lie algebra~$\mathfrak{g}$ of~$G$ has basis~$\{Y_\alpha\}_{\alpha\in A}$. 

We will consider an environment or bath~$\cH_E$ with associated Lie algebra $\cB(\cH_E)$ consisting of all bounded operators on~$\cH_E$. Let $\Xorig:[0,T]\rightarrow\mathfrak{g}\otimes\cB(\cH_E)$ be a time-dependent generator describing system-bath interactions, i.e., it is of the form
\begin{equation}\label{eq:systembathcoupling}
\Xorig(t)=\sum_{\alpha\in A} Y_\alpha \otimes B_\alpha(t)
\end{equation}
for some elements $B_\alpha(t)\in\cB(\cH_E)$. We will further assume that these functions are analytic with expansion
\begin{eqnarray}
B_\alpha(t)=\sum_{r=0}^\infty b_{\alpha,r} t^r\quad\textrm{ for each }\alpha\in A\ .\label{eq:analyticexpansionmodeops}
\end{eqnarray}
Let $\xorig:[0,T]\rightarrow \mathsf{GL}(\cH_S\otimes\cH_E)$ define the original (uncontrolled) evolution generated by~$\Xorig$, i.e., it is defined by~$\frac{d}{dt}\xorig(t)=\Xorig(t) \xorig(t)$ for~$t\in(0,T)$ and~$\xorig(0)=I_S\otimes I_E$. We note that we use upper case letters to denote Lie algebra elements and lower case letters to denote Lie group elements throughout this section.
Consider the adjoint action $\ad:G\rightarrow\mathsf{GL}(\mathfrak{g})$ defined as
\begin{equation}
\ad(x)(Y)=xYx^{-1}\quad\textrm{ for all }x\in G\textrm{ and }Y\in\mathfrak{g}\ .
\end{equation}
We fix a family $\{x_\beta\}_{\beta\in B}\subset G$ of group elements which act diagonally in the chosen basis of~$\mathfrak{g}$ in the sense that
\begin{equation}
\ad(x_\beta)(Y_\alpha)=(-1)^{\langle\alpha,\beta\rangle}Y_\alpha \label{eq:adxinvidentity}
\end{equation}for all $(\alpha,\beta)\in A\times B$. We note that in the cases of interest, $\langle\alpha,\beta\rangle$ is the symplectic inner product modulo 2 as defined in the main article. 
Let us also assume that each $x_\beta$ has an infinitesimal generator~$X_\beta\in\mathfrak{g}$ up to a complex phase $e^{i\varphi}$ where $\varphi\in\mathbb{R}$ in the sense that 
\begin{equation}
x_\beta=e^{i\varphi} e^{X_\beta}\qquad\textrm{ for each }\beta\in B\ . \label{eq:generatorxX}
\end{equation}
In the case where $G$ is a real Lie group, the phase $e^{i\varphi}$ should be replaced by $\pm 1$. We note that the adjoint action $\ad(x_\beta)$ does not depend on this phase.

Consider the stroboscopic application of pulses $x_{\beta(\lambda)}\in G$ to the system at times $t_\lambda=T\Delta_\lambda\in [0,T]$, for  each $\lambda$ belonging to a finite set~$\Lambda$. That is, the function $\beta:\Lambda\rightarrow B$ specifies which pulse is applied at time~$t_\lambda$.
Let us define the control evolution $\xcontrol:[0,T]\rightarrow \mathsf{GL}(\cH_S\otimes\cH_E)$ as the product of all pulses applied up some time where the order of factors is defined by the ordering $t_\lambda\le t_{\lambda'}$ of times, that is (again, up to a phase~$e^{i\varphi}$ or~$\pm 1$)
\begin{equation}\label{eq:prodpulses}
\xcontrol(t)=\prod_{\lambda: \,t_\lambda\leq t} x_{\beta(\lambda)} \otimes I_{\cH_E}\ .
\end{equation} We will assume that applying all pulses up to time~$T$ amounts to the identity operation (again up to a phase).
We are interested in the evolution $\xres:[0,T]\rightarrow \mathsf{GL}(\cH_S\otimes\cH_E)$ that results if we apply the pulse~$x_{\beta(\lambda)}$ (instantaneously) at time $t_\lambda$ for each $\lambda\in \Lambda$ and let the system evolve freely under $\Xorig(t)$ at all other times.

\subsection{Toggling frame, Dyson expansion and sufficient decoupling criteria}\label{app:decouplingcriteria}
To analyze the evolution $\xres:[0,T]\rightarrow \mathsf{GL}(\cH_S\otimes\cH_E)$, it is convenient to change into the toggling frame with evolution $\xtf:[0,T]\rightarrow\mathsf{GL}(\cH_S\otimes\cH_E)$ defined as~$\xtf(t)= \xcontrol(t)^{-1} \xres(t)$ for all $t\in [0,T]$. Its generator is then given by~$\Xtf(t)= \xcontrol(t)^{-1} \Xorig(t) \xcontrol(t)$ and hence independent of the phases in $x_\beta$ and $\xcontrol(t)$.
Eq.~\eqref{eq:adxinvidentity} implies that the toggling generator is of the simple form
\begin{eqnarray}
\Xtf(t)&= \sum_{\alpha\in A} F_\alpha (t/T)\, Y_\alpha \otimes B_\alpha (t)\label{eq:togglingframexplicitgen}
\end{eqnarray} 
where the function  $F_\alpha: [0,1]\to \{-1,1\}$ for $\alpha\in A$ is defined as
\begin{eqnarray}
F_\alpha (\tau)= (-1)^{\sum_{\lambda: \Delta_\lambda\leq \tau} \langle\alpha,\beta(\lambda)\rangle}\quad\textrm{for }\tau\in [0,1]\ .\label{eq:falphafamilydef}
\end{eqnarray}
Expanding the toggling frame evolution $\xtf(T) $ generated by~\eqref{eq:togglingframexplicitgen} in a Dyson series gives 
\begin{equation}\label{eq:Dyson_expansion_toggling}
\xtf(T)= \sum_{s=0}^\infty \sum_{\vec{\alpha}\in A^s} \prod_{k=1}^s Y_{\alpha_k}  \otimes \sum_{\vec{r}=0}^\infty\prod_{k=1}^s b_{\alpha_k,r_k} \cF^{\vec{r}}_{\vec{\alpha}}\ T^{s+\sum_{k=1}^s r_k} 
\end{equation}
for $\vec{\alpha}=(\alpha_1,\ldots,\alpha_s)$, $\vec{r}=(r_1,\ldots,r_s)$ where we have defined the scalars
\begin{equation}\label{eq:def_F}
\cF^{\vec{r}}_{\vec{\alpha}}=\int_{0}^1 d\tau_s\cdots \int_0^{\tau_{2}} d\tau_1\, \prod_{k=1}^s F_{\alpha_k}(\tau_k)\tau^{r_k}_k \;. 
\end{equation}
We note that after time $T$, the toggling frame evolution~$\xtf(T)$ is equal to the resulting evolution~$\xres(T)$ up to a phase. 
The $N$-th order term in $T$ is now given by Eq.~\eqref{eq:Dyson_expansion_toggling}.
This gives the following statement (cf.~\cite{jiang,kuolidar11}), where we write $A\not\propto B$ if $A$ is not a scalar multiple of $B$.

\begin{theorem}[Decoupling criterion]\label{thm:decoupling_criterion}
	Consider $N\in\mathbb{N}$ and the scalars $\cF^{r_1,\ldots,r_s}_{\alpha_1,\ldots,\alpha_s}$ defined by~\eqref{eq:def_F}. 
	For all $s\in \mathbb{N}$, $r_1,\ldots,r_s\in \mathbb{N}_0$, and $\alpha_1, \ldots, \alpha_s\in A$ assume that
	\begin{equation}
	\cF_{\vec{\alpha}}^{\vec{r}}(\{F_\alpha\}) =0\quad \textrm{ if }
	\begin{cases}
	s+\sum_{k=1}^s r_k\le N\textrm{ and }\\
	Y_{\alpha_s}\cdots Y_{\alpha_1} \not\propto I_S\ .
	\end{cases}	\label{eq:conditionsgeneral}
	\end{equation}
	Then there is an operator $\tilde{B}\in \mathcal{B}(\cH_E)$ such that
	\begin{equation}
	\| \xres(T)-I_S\otimes \tilde{B}\|= O(T^{N+1}) \; .
	\end{equation} 
\end{theorem}
The constant in $O(T^{N+1})$ will depend on the norm of the original environment operators $B_\alpha$ and the resulting environment operator $\tilde{B}$ will include the phase relating~$\xtf(T)$ and~$\xres(T)$.
We will also require a weaker form of decoupling, where the system-environment interaction is reduced to a particular form (specified by a single basis element $Y_\gamma\in \mathfrak{g}$) up to order~$N$. This also follows immediately from~\eqref{eq:Dyson_expansion_toggling}. 
\begin{corollary}[Modified decoupling/homogenization criterion]\label{thm:mod_decoupling_criterion}
	Let $\gamma\in A$ be fixed. 
	For all $s\in \mathbb{N}$, $r_1,\ldots,r_s\in \mathbb{N}_0$, and $\alpha_1, \ldots, \alpha_s\in A$ assume that
	\begin{equation}
	\cF_{\vec{\alpha}}^{\vec{r}}(\{F_\alpha\}) =0\quad \textrm{ if} \quad
	\begin{cases}
	s+\sum_{k=1}^s r_k\le N\textrm{ and }\\
	Y_{\alpha_s}\cdots Y_{\alpha_1} \not\propto \{I_S,Y_\gamma\} \ .
	\end{cases}	\label{eq:conditionsgeneralmod}
	\end{equation}
	Then there are operators $\tilde{B}_1,\tilde{B}_2\in \mathcal{B}(\cH_E)$ such that 
	\begin{equation}
	\|\xres(T)- I_{S}\otimes\tilde{B}_1-Y_\gamma\otimes\tilde{B}_2\|= O(T^{N+1})\ .
	\end{equation}
\end{corollary}

We emphasize that the scalars $\cF^{r_1,\ldots,r_s}_{\alpha_1,\ldots,\alpha_s}$ depend on the family of functions $\{F_{\alpha}\}_{\alpha\in A}$ which itself is defined by the tuple $(\{t_\lambda\}_\lambda,\beta)$ in Eq.~\eqref{eq:falphafamilydef}. 
Hence the decoupling properties of a given pulse sequence are captured by the corresponding functions~$F_\alpha$.

In Section~\ref{app:udd}, we discuss two known examples of functions~$F_\alpha$ where the conditions of Theorem~\ref{thm:decoupling_criterion} are met, the UDD and the NUDD scheme. 

\subsection{Qubit decoupling revisited: UDD and NUDD schemes}\label{app:udd}
Let us consider a system $\cH_S=(\mathbb{C}^2)^{\otimes m+1}$ of $m+1$ qubits -- labeled from 0 to $m$ -- that interacts with an environment $\cH_E$
via the original Hamiltonian 
\begin{align}\label{eq:NUDDHamiltonianorig}
&\Horig(t)=\sum_{\alpha\in (\mathbb{Z}_2^2)^{m+1}} \sigma_{\alpha} \otimes B_{\alpha} (t)\ 
\end{align} where $\sigma_\alpha$ denote multi-qubit Pauli operators 
for any sequence $\alpha=(a_0,a_1,\ldots,a_m)\in (\mathbb{Z}_2^2)^{m+1}$ such that $a_i\in \mathbb{Z}_2^2$ for $i=0,\ldots,m$. The environment operators $B_{\alpha}(t)$ are assumed to be time-dependent and analytic with series expansion that satisfies Eq.~\eqref{eq:analyticexpansionmodeops}.

The Hamiltonian~$\Horig(t)$ falls into the above framework. In particular, the adjoint action of Lie group elements~\eqref{eq:adxinvidentity} is 
\begin{align}\label{eq:rel_Salpha_NUDD_appendix}
\sigma_{\beta}^{-1}\sigma_{\alpha} \sigma_{\beta}=(-1)^{\langle\alpha, \beta\rangle} \sigma_{\alpha} \quad \textrm{for } \alpha,\beta\in (\mathbb{Z}_2^2)^{m+1}.
\end{align}

To eliminate decoherence induced by~\eqref{eq:NUDDHamiltonianorig}, we consider the nested Uhrig DD (NUDD) sequence~\cite{mukhtar10,wangliu11,jiang}. Historically this scheme was deduced from the Uhrig DD scheme~\cite{uhrig}. First, a single qubit was considered, where the quadratic DD (QDD) sequence was introduced by West et al.~\cite{westetal10} to generalize the Uhrig sequence to arbitrary system-environment interactions; proofs of its validity were subsequently provided by Wang and Liu~\cite{wangliu11} as well as Kuo and Lidar~\cite{kuolidar11}. 
Mukhtar et al.~\cite{mukhtar10} extended this scheme to protect unknown two-qubit states from decoherence. The generalization to $n$ qubits, i.e., the NUDD sequence, was shown to be universal in~\cite{wangliu11} where the authors considered even (but potentially different) decoupling orders in every nesting level.
An alternate proof~\cite{jiang} showed universality of the NUDD sequence for arbitrary decoupling order (even or odd, but the same in every level).
We note that this discussion of the (multi-)qubit DD is by no means exhaustive, but covers those aspects
which are directly pertinent to our work. For further work e.g., on particular noise models (such as~\cite{violalloyd98,violaknilllloyd99,shiokawalidar04,santosviola05,pasiniuhrigpra10}) or the discussion of finite-width pulses, we refer to the literature. 

Let us introduce the label $\lambda\in \{0,1,\ldots,N\}^{2(m+1)}$ as
\begin{align}
\lambda=(\ell_0,\ldots,\ell_{2m+1}) \quad\text{ where } \quad \ell_k\in \{0,1,\ldots,N\} \ .
\end{align} 
In order to define the pulse times $t_\lambda=T\Delta_\lambda$, let
\begin{align}\label{eq:Uhrigtimesapp}
\Delta_j=\sin^2\frac{j \pi }{2(N+1)}
\end{align} for $j= 1,2,\ldots,N$ be the Uhrig DD times. The intuition behind $\Delta_\lambda$ being an $(2m+2)$-fold concatenation of the Uhrig pulse times $\Delta_j$ is the following: First, for times $\Delta_j$, we divide $[0,1]$ into $N+1$ intervals $\Delta_j$; this will be called outermost level. Then $\Delta_{(\ell_0,\ell_1)}$ on the next level, is obtained by subdividing each of these $N+1$ intervals again into $N+1$ parts by $\Delta_j$. This concatenating procedure is recursively repeated.

Formally, we set $\Delta_0=0$, $\Delta_{N+1}=1$ and starting with $d_j=\Delta_j$,
recursively introduce the quantity
\begin{equation}\label{eq:def_dlambda}
d_{(\ell_0,\ell_1,\ldots,\ell_k)}= \Delta_{\ell_k}+(\Delta_{\ell_k+1}-\Delta_{\ell_k})d_{(\ell_0,\ell_1,\ldots, \ell_{k-1})}
\end{equation}
for $\ell_0,\ell_1,\ldots,\ell_k\in\{0,1,\ldots,N+1\}$ and $k=1,\ldots,2m+1$. Recalling Eq.~\eqref{eq:def_dlambda} the pulse times are 
\begin{equation}\label{eq:NUDD_times}
	t^{\textsf{NUDD}}_\lambda=T\Delta_\lambda
\end{equation} where for $\lambda=(\ell_0,\ldots,\ell_{2m+1})$ we set
\begin{align}
    \Delta_\lambda=\begin{cases}
        1 &\textrm{if } \lambda=(0,\ldots,0)\\
        d_{\lambda}& \textrm{if } \ell_0\neq 0\\
        d_{\lambda'(r)} &\textrm{if } \ell_r\neq0\textrm{ for }r\in \{1,\ldots,2m+1\}\\
        & \quad\textrm{and }\ell_0=\cdots=\ell_{r-1}= 0
    \end{cases}\quad \label{eq:NUDD_Deltalambda}
\end{align} where $\lambda'(r)$ differs from $\lambda$ only in the components $\ell'_{r-1}=N+1$ and $\ell'_r=\ell_r -1$.

If $N$ is even, the NUDD control pulses are
(up to factors~$\pm 1$ or $\pm i$) given by
\begin{align}\label{eq:NUDDpulesNeven}
U^{\textsf{NUDD}}_\lambda=
\begin{cases}
\mathsf{id} & \textrm{if } \lambda=(0,\ldots,0)\\
(\sigma_z)_k & \textrm{if } \ell_{2k}\neq 0\textrm{ for }k\in \{0,\ldots,m\}\\
  & \quad\textrm{and }\ell_i=0 \textrm{ for }i\le2k-1\\
(\sigma_x)_k & \textrm{if } \ell_{2k+1}\neq 0 \textrm{ for }k\in \{0,\ldots,m\}\\
  & \quad\textrm{and } \ell_i=0 \textrm{ for }i\le2k
\end{cases}\quad
\end{align}
where $(\sigma_{x,y,z})_k=I^{\otimes k}\otimes \sigma_{x,y,z}\otimes I^{\otimes m-k} $ and $k=0,\ldots,m$. We note that a complex phase of the unitary pulse operators does not have any effect on the decoupling analysis, since we are only interested in terms of the form $( U^{\textsf{NUDD}}_\lambda)^{-1} \sigma_\alpha U^{\textsf{NUDD}}_\lambda$. In slight abuse of notation, we will therefore omit these phases in the definition of the pulses. If~$N$ is odd, then the pulses are defined slightly differently, by taking~\eqref{eq:NUDDpulesNeven} and replacing
\begin{align}\label{eq:NUDDpulesNodd}
\begin{matrix}
\mathsf{id}&\mapsto&\prod_{j=0}^{m}(\sigma_y)_j\\
(\sigma_z)_k&\mapsto&(\sigma_z)_k \prod_{j=1}^{k-1}(\sigma_y)_j\\
(\sigma_x)_k&\mapsto&\prod_{j=0}^{k}(\sigma_y)_j
\end{matrix} \ .
\end{align}

The toggling frame Hamiltonian is given by
\begin{equation}
\Htf(t)=\sum_{\alpha} F^{\mathsf{NUDD}}_{\alpha}(t/T)S_\alpha\otimes B_\alpha(t)
\end{equation}
for the family of functions $F^{\mathsf{NUDD}}_{\alpha}$ defined below. Due to $\sigma_x^2=\sigma_y^2=\sigma_z^2=I$ and the commutation relations between Pauli matrices, we have $\sigma_{\alpha_1}\cdots \sigma_{\alpha_s}\propto \sigma_{\oplus_{k=1}^s\alpha_k}$ up to factors $\pm1$ or $\pm i$. Then the order $N$ decoupling property of the NUDD sequence follows from Theorem~\ref{thm:decoupling_criterion} and the statement of the following lemma.
\begin{lemma}[\cite{jiang}]\label{lem:NUDD_decoupling_condition}
	For $N\in\mathbb{N}$, $\lambda\in \{0,1\ldots,N\}^{2m+2}$ and $\alpha\in (\mathbb{Z}^2_2)^{m+1}$, let $F^{\mathsf{NUDD}}_{\alpha}:[0,1]\to \{-1,1\}$ be defined as
	\begin{equation}\label{eq:def_F_NUDD}
	F^{\mathsf{NUDD}}_{\alpha}(\tau)= (-1)^{{\alpha} \cdot\lambda} \quad \text{ for } \;\tau  \in (\Delta_\lambda,\Delta_{\lambda_+}]
	\end{equation}
	where $\lambda_+$ labels the pulse following the one with label~$\lambda$, $\Delta_\lambda$ is defined in~\eqref{eq:NUDD_Deltalambda} and~$\cdot$ denotes the scalar product. 
	Then for all 
    $s\in \mathbb{N}$, $\vec{r}=(r_1,\ldots,r_s)\in \mathbb{N}^s_0$, and $\vec{\alpha}=(\alpha_1, \ldots, \alpha_s)\in \big((\mathbb{Z}^2_2)^{m+1}\big)^s$,
	\begin{align}
        \cF_{\vec{\alpha}}^{\vec{r}}(\{F^\mathsf{NUDD}_\alpha\}) =0\quad\textrm{if }
    \begin{cases}
        s+\sum_{j=1}^s r_j\le N\textrm{ and }\\
    \bigoplus_{k=1}^s \alpha_k \neq (0,\ldots,0)
    \end{cases}
    \end{align} where $\oplus$ denotes entriwise addition modulo two.
\end{lemma}

Here we present another specific family of functions that satisfies the conditions of Theorem~\ref{thm:decoupling_criterion}. This is associated with the UDD sequence introduced by Uhrig~\cite{uhrig}. It can be regarded as special case of the NUDD sequence (one qubit and one concatenation level). Here $\Horig(t)=I\otimes B_0(t)+\sigma_z\otimes B_1(t)$ -- involving only $\sigma_z$-Pauli operators on the system. The pulses $\sigma_x$ are applied at times $t^\textsf{UDD}_j$ from Eq.~\eqref{eq:Uhrigtimes}.
One can show that the toggling frame Hamiltonian is $\Htf(t)=F_0^{\mathsf{UDD}}(t/T) I\otimes B_0(t)+ F_1^{\mathsf{UDD}}(t/T)\sigma_z\otimes B_1(t)$ for the functions $F_0^{\mathsf{UDD}},F_1^{\mathsf{UDD}}$ defined below. The universality of the UDD scheme, proved in~\cite{YangLiu08} (see also~\cite{uhriglidar}), then relies on the following lemma.
\begin{lemma}\label{lem:UDDfunctions}
	For $N\in\mathbb{N}$ and $\alpha\in \mathbb{Z}_2$, let $F_\alpha^{\mathsf{UDD}}:[0,1]\rightarrow\{-1,1\}$ be defined by
	\begin{align}
	F^{\mathsf{UDD}}_0(\tau)&=1 		&&\textrm{ for all }\tau\in [0,1]\qquad\textrm{ and }\label{eq:f0def}\\
	F^{\mathsf{UDD}}_1(\tau)&=(-1)^j	&&\textrm{ for all }\tau\in (\Delta_j,\Delta_{j+1}]\ ,\label{eq:f1def}
	\end{align} where $\Delta_j$ is defined by Eq.~\eqref{eq:Uhrigtimesapp}. 
	Then for all 
    $s\in \mathbb{N}$, $\vec{r}=(r_1,\ldots,r_s)\in \mathbb{N}^s_0$, and $\vec{\alpha}=(\alpha_1, \ldots, \alpha_s)\in \mathbb{Z}^s_2$,
	\begin{align}
        \cF_{\vec{\alpha}}^{\vec{r}}(\{F^\mathsf{UDD}_\alpha\}) =0\quad\textrm{if }
    \begin{cases}
        s+\sum_{j=1}^s r_j\le N\textrm{ and }\\
    \bigoplus_{k=1}^s \alpha_k \neq 0
    \end{cases}	
    \end{align} where $\oplus$ denotes addition modulo two.
\end{lemma}

We reuse this lemma in the context of bosonic decoupling.

\section{Proof of decoupling}\label{app:proofdec}

We consider decoupling of matrices of the form
\begin{equation}\label{eq:Xbosdecapp}
\Xorig(t)=\begin{pmatrix}
X_{SS}(t) &X_{SE}(t)\\
X_{ES}(t)& X_{EE}(t)
\end{pmatrix}
\end{equation}
where the functions $X_{BC}(t)$ for $B,C\in \{S,E\}$ are analytic with series expansions
\begin{equation}
X_{BC}(t)=\sum_{r=0}^\infty X_{BC,r} t^{r}\quad\textrm{for } B,C\in \{S,E\}\ .\label{eq:analyticseriesexpansion}
\end{equation}
Given a bosonic pulse sequence with pulses~$-I_S$ applied at for now not further specified times $t_j$, we analyze the Dyson expansion of the toggling frame evolution 
\begin{equation}\label{eq:StfDyson}
\Stf(T)=\sum_{s=0}^\infty \int_0^T  dt_s\cdots \int_0^{t_{2}}dt_1\  
\prod_{k=1}^s\Xtf(t_k) \ . 
\end{equation} To show $N$-th order decoupling, we prove that $\Stf(T)$ is of direct sum form $S_S\oplus S_E$ (for matrices $S_S\in \mathbb{R}^{2n_S\times 2n_S}$ and $S_E\in \mathbb{R}^{2n_E\times 2n_E}$) up to order $N$ in $T$, i.e.,~that the off-diagonal terms of~\eqref{eq:StfDyson} --$\left(\Stf(T)\right)_{SE}$ and $\left(\Stf(T)\right)_{ES}$ -- vanish up to $O(T^{N+1})$.

\begin{lemma}\label{lem:maindecouplingdirectsum}
	Consider a pulse sequence defined by a piecewise constant function $\sigma:[0,1]\rightarrow \{-1,1\}$ that satisfies $\sigma(0)=1$ and changes its sign when the pulse $-I_S$ is applied. Suppose for all~$s\in\mathbb{N}$, $r_1,\ldots,r_s\in\mathbb{N}_0$ and $\gamma_1,\ldots,\gamma_s\in \mathbb{Z}_2$ the function~$\sigma$ satisfies
	\begin{equation}
	\cF_{\gamma_1,\ldots,\gamma_s}^{r_1,\ldots,r_s}(\sigma)=0\quad
	\textrm{ if }\ \begin{cases}
	s+\sum_{k=1}^s r_k\leq N\textrm{ and }\\
	\bigoplus_{k=1}^s \gamma_k=1\ 
	\end{cases}	\label{eq:decouplingsufficientcond}
	\end{equation}
	where $\oplus$ denotes addition modulo 2 and we have defined
	\begin{equation}\label{eq:scalarscFbosdec}
	\cF_{\gamma_1,\ldots,\gamma_s}^{r_1,\ldots,r_s}(\sigma)=\int_0^1  d\tau_s\cdots \int_0^{\tau_{2}}d\tau_1 \prod_{k=1}^s\sigma(\tau_k)^{\gamma_k}\tau_k^{r_k}\ .
	\end{equation}
	Then the toggling frame evolution is of the form
	\begin{equation}
	\|\Stf(T)- S_{S}\oplus S_{E}\|=O(T^{N+1})\label{eq:errortermstfsse}
	\end{equation}
	for operators $S_{S}$ and $S_{E}$ acting on the system and environment only. 
\end{lemma}
\begin{proof}
	Compute the expression
	\begin{equation}
	\left(\prod_{k=1}^s\Xtf(t_k)\right)_{SE}=\sum_{(B,C)\in\cV_s(S,E)} \prod_{k=1}^s \Xtf_{B_kC_k}(t_k)\quad \label{eq:firstidentityproofcompos}
	\end{equation} inside the integrals in the Dyson expansion of~\eqref{eq:StfDyson}.
	Here we sum over the set $\cV_s(S,E)$ of sequences
	$(B,C)=(B_1,\ldots,B_s,C_1,\ldots,C_s)\in \{S,E\}^{2s}$
	such that 
	\begin{align}
	B_s&=S\\
	C_1&=E\qquad\textrm{ and }\label{eq:chainconditionproduct}\\ 
	C_{k+1}&=B_{k}\qquad\textrm{for all }k\in [s-1]\ .
	\end{align}
	We define
	\begin{equation}
	\gamma(B,C)=
	\begin{cases}
	0 & (B,C)\in \{(S,S),(E,E)\}\\
	1& (B,C)\in \{(S,E),(E,S)\}
	\end{cases}
	\end{equation}
	and note that Eqs.~\eqref{eq:chainconditionproduct} imply the identity
	\begin{equation}
	\bigoplus_{k=1}^s \gamma(B_k,C_k)=1 \quad\textrm{for all }
	(B,C)\in\cV_s(S,E) .\label{eq:vanishingconditionajbj}
	\end{equation}
	With $\Xtf(t)$ given by 
	\begin{equation}
	\Xtf(t)=\begin{pmatrix}
	X_{SS}(t)& \sigma(t/T)X_{SE}(t)\\
	\sigma(t/T)X_{ES}(t) & X_{EE}(t)
	\end{pmatrix}\ ,
	\end{equation} the expression~\eqref{eq:firstidentityproofcompos} becomes
	\begin{equation}
	\left(\prod_{k=1}^s\Xtf(t_k)\right)_{SE}
	=\sum_{(B,C)\in\cV_s(S,E)}\prod_{k=1}^s\sigma\big(\tfrac{t_k}{T}\big)^{\gamma_k}\, X_{B_kC_k}(t_k) 
	\end{equation} where we write $\gamma_k=\gamma(B_k,C_k)$.
	Inserting this and the analytic expansions~\eqref{eq:analyticseriesexpansion} into the upper off-diagonal part of the Dyson expansion~\eqref{eq:StfDyson} gives
	\begin{align}
	\left(\Stf(T)\right)_{SE}=\sum_{s=0}^{\infty }
	&\sum_{(B,C)\in\cV_s(S,E)} \sum_{r_1,\ldots,r_s}  \prod_{k=1}^s  X_{B_kC_k,r_k} \cdot\\
	& \qquad \cdot \cF_{\gamma_1,\ldots,\gamma_s}^{r_1,\ldots,r_s} (\sigma)\ T^{s+\sum_{j=1}^sr_j} \label{eq:seexpression}
	\end{align}
	where $\cF(\sigma)$ is defined by~\eqref{eq:scalarscFbosdec}. 
	With property~\eqref{eq:vanishingconditionajbj} of the sequences in $\cV_s(S,E)$ and the assumptions Eq.~\eqref{eq:decouplingsufficientcond} we conclude that
	\begin{equation}
	\cF_{\gamma(B_1,C_1),\ldots,\gamma(B_s,C_s)}^{r_1,\ldots,r_s}=0
	\end{equation}
	whenever $s+\sum_{k=1}^s r_k\leq N$ and $(B,C)\in\cV_s(S,E)$. 
	Thus Eq.~\eqref{eq:seexpression} gives
	\begin{equation}
	\left(\Stf(T)\right)_{SE}=O(T^{N+1})\ .
	\end{equation}
	Analogous reasoning yields $\left(\Stf(T)\right)_{ES}=O(T^{N+1})$, hence the claim follows.
\end{proof}

Now consider the concrete scheme from Theorem~\ref{thm:decoupling}, where the pulse $-I_S$ is applied at the Uhrig times $t^{\mathsf{UDD}}_j=\Delta_j T$ from~\eqref{eq:Uhrigtimesapp}.  Let $F^{\mathsf{UDD}}_0,F^{\mathsf{UDD}}_1:[0,1]\rightarrow \{-1,1\}$ be the functions introduced in Lemma~\ref{lem:UDDfunctions}. They satisfy
\begin{align}\label{eq:FUDDsigma}
F^{\mathsf{UDD}}_\gamma(\tau)=\sigma(\tau)^\gamma\quad\textrm{ for all }\tau\in [0,1],\gamma\in\mathbb{Z}_2\ .
\end{align} since $F^{\textsf{UDD}}_1(\tau)=\sigma(\tau)$ and $F^{\textsf{UDD}}_0(\tau)=1$ is the constant function.
	Recalling Lemma~\ref{lem:maindecouplingdirectsum}, 
	it suffices to show that the function~$\sigma$ satisfies the condition
	defined by~\eqref{eq:decouplingsufficientcond}.   
	Inserting~\eqref{eq:FUDDsigma}, this condition takes the form that for all $s\in\mathbb{N}$, $r_1,\ldots,r_s\in\mathbb{N}_0$ and
	$\gamma_1,\ldots,\gamma_s\in \mathbb{Z}_2$:
	\begin{align}
	\int_0^1  d\tau_s\cdots \int_0^{\tau_{2}}d\tau_1 \prod_{k=1}^s F^{\mathsf{UDD}}_{\gamma_k}(\tau_k)\tau_k^{r_k}&=0
	\end{align}
	if $s+\sum_{k=1}^s r_k\leq N$ and $\bigoplus_{k=1}^s \gamma_k\neq 0$. 
	According to Lemma~\ref{lem:UDDfunctions} the functions $F^{\mathsf{UDD}}_0$, $F^{\mathsf{UDD}}_1$ have this property.
Lemma~\ref{lem:maindecouplingdirectsum} thus implies that the toggling frame evolution~$\Stf(T)$ is decoupled up to order~$N$ and (by $\Sres(T)=\pm\Stf(T)$) the resulting evolution~$\Sres(T)$ as well. This proves Theorem~\ref{thm:decoupling}.

\section{A bound on sufficient decoupling rates}\label{app:suffrate}
Application of our decoupling scheme requires a DD control rate~$1/T$ sufficiently large compared to the energy scale set by the uncontrolled system-bath evolution. To estimate this in more detail, we give a bound on the constant appearing in the error~\eqref{eq:errortermstfsse}. We consider the time-independent case and show the following:

\begin{lemma}\label{lemma:sufficientdecrate}
	Let $\Xorig=\begin{pmatrix}
	X_{SS}&X_{SE}\\
	X_{SE}&X_{EE}
	\end{pmatrix}\in\mathfrak{sp}(2(n_S+n_E))$ be a time-independent generator and define
	\begin{align}
	J_0=\|X_{EE}\| \ \textrm{ and }\  J_z=\|X_{SS}\|+\|X_{SE}\| \ .
	\end{align}
	Then applying $N$ pulses at UDD times as above results in the evolution
	\begin{align}
	\Sres(T)=\begin{pmatrix}
	\Sres(T)_{SS}&\Sres(T)_{SE}\\
	\Sres(T)_{ES}&\Sres(T)_{EE}
	\end{pmatrix}
	\end{align} where
	\begin{align}
	\|\Sres(T)_{ES}\|,\|\Sres(T)_{SE}\|\leq\sum_{s=N+1}^\infty \frac{(J_0+J_z)^{s}}{s!}\cdot T^{s}\ . 
	\end{align}
	In particular, there are $S_S\in \mathbb{R}^{2n_S\times 2n_S}$ and $S_E\in\mathbb{R}^{2n_E\times 2n_E}$ such that
	we have the bound
	\begin{align}
	\|\Sres(T)-S_S\oplus S_E\|\leq \frac{e\sqrt{2}((J_0+J_z)T)^{N+1}}{(N+1)!}\ \label{eq:sresss}
	\end{align}
	if $(J_0+J_z)T\leq 1$. 
\end{lemma}

This bound is similar in spirit to the 
analysis of Uhrig decoupling in a model of a single spin with pure dephasing~\cite{uhriglidar}, i.e., $H=I_S\otimes B_0+\sigma_z\otimes B_z$. Here $B_z$ is not necessarily traceless and may involve system-only evolution terms.  
It was shown in~\cite{uhriglidar} that the  error term takes the form~$O(1/((N+1)!)(J_0+J_z)^{N+1}T^{N+1})$, 
where $J_0=\|B_0\|$ and $J_z=\|B_z\|$.  
\begin{proof}
	For the time-independent case, we have $X_{BC,r}=0$ for $r>0$, hence Eq.~\eqref{eq:seexpression} reduces to
	\begin{align}
	\big(\Stf(T)\big)_{SE}=\sum_{s=0}^{\infty } &
	\sum_{(B,C)\in\cV_s(S,E)} \prod_{k=1}^s  X_{B_kC_k}\\& \cdot \cF_{\gamma_1,\ldots,\gamma_s}^{0,\ldots,0} (\sigma)\  T^{s} \label{eq:StfSEimeindependent}
	\end{align}
	Observe that the first $N$ terms of this series vanish due to the property
	\begin{equation}
	\cF_{\gamma_1,\ldots,\gamma_s}^{0,\ldots,0}(\sigma)=0\quad
	\textrm{ if }\ \begin{cases}
	s\leq N\textrm{ and }\\
	\bigoplus_{k=1}^s \gamma_k=1\ 
	\end{cases}	\label{eq:FUDDtimindependent}
	\end{equation} of the UDD sequence. For the terms with $s\ge N+1$ we use that $\|X_{SE}\|=\|X_{ES}\|$ (which follows from the definition of $X=AJ$ and the fact that $A$ is symmetric) to bound
	\begin{align}
	\sum_{(B,C)\in\cV_s(S,E)} &\| \prod_{k=1}^s  X_{B_kC_k}\| \label{eq:XSEproduct}\\
	\leq &\sum_{\substack{k_1,k_2,k_3\\ k_1+k_2+k_3=s}}\frac{\|X_{SE}\|^{k_1}\|X_{SE}\|^{k_2}\|X_{SE}\|^{k_3}}{k_1!k_2!k_3!} \\
	= &\ (\|X_{SS}\|+\|X_{SE}\|+\|X_{EE}\|)^s\ .
	\end{align}
	We also bound
	\begin{align}
	|\cF_{\gamma_1,\ldots,\gamma_s}^{0,\ldots,0}|\leq \int_0^1  d\tau_s\cdots \int_0^{\tau_{2}}d\tau_1 \leq
	\frac{1}{s!}\ .
	\end{align}
	Inserting this, Eqs.~\eqref{eq:XSEproduct} and~\eqref{eq:FUDDtimindependent} into~\eqref{eq:StfSEimeindependent} results in
	\begin{align}
	\|&\left(\Stf(T)\right)_{SE}\|\\
	&\leq \sum_{s=N+1}^{\infty }
	\sum_{(B,C)\in\cV_s(S,E)}\| \prod_{k=1}^s  X_{B_kC_k}\| \cdot |\cF_{\gamma_1,\ldots,\gamma_s}^{0,\ldots,0} (\sigma)|\ T^{s}\\
	&\leq \sum_{s=N+1}^{\infty }\frac{1}{s!}(\|X_{SS}\|+\|X_{SE}\|+\|X_{EE}\|)^s \ T^{s}\ . \label{eq:}
	\end{align} As $\Stf(T)=\Sres(T)$, this proves the first statement. 
	Under the assumption $(J_0+J_z)T\leq1$ then we have
	\begin{equation}
	\|\left(\Stf(T)\right)_{SE}\|\leq ((J_0+J_z)T)^{N+1}\sum_{k=0}^{\infty }\frac{1}{(N+1+k)!}\ \label{eq:stfseex}
	\end{equation}
	We note that $R_{N}:=\sum_{k=0}^{\infty }\frac{1}{(N+1+k)!}$ is the $N$-th remainder term in the Taylor series of the exponential function around~$0$, i.e., 
	\begin{align}
	\exp(1)&=\sum_{k=0}^{N}\frac{1}{k!}+R_{N}
	\end{align}
	and can thus be expressed by the Lagrange form
	\begin{align}
	R_{N}&=\frac{\exp^{(N+1)}(\xi)}{(N+1)!}1^{N+1}\quad\textrm{ for some }\xi\in [0,1]\ .
	\end{align}
	We conclude that
	\begin{align}
	|R_{N}|\leq \frac{e}{(N+1)!}
	\end{align}
	Inserting this into~\eqref{eq:stfseex} gives 
	\begin{equation}
	\|\left(\Stf(T)\right)_{SE}\|\leq \frac{e\ ((J_0+J_z)T)^{N+1}}{(N+1)!}\ .
	\end{equation}Since the same bound holds for $\|\left(\Stf(T)\right)_{ES}\| $ and $\|\Sres(T)-S_S\oplus S_E\|^2= \|\left(\Stf(T)\right)_{SE}\|^2+\|\left(\Stf(T)\right)_{ES}\|^2$ 
	for $S_S=\Sres(T)_{SS}$ and $S_E=\Sres(T)_{EE}$,	we obtain Eq.~\eqref{eq:sresss}.
\end{proof}
We note that the more refined bounds obtained in~\cite{uhriglidar} expressed in terms of the ``odd part of the bounding series'' (cf.~\cite[Eq.~(12)]{uhriglidar}) appear to be less straightforward to generalize to the bosonic setting.
One way of improving the above estimate may use the fact that only terms with $k_2$ odd appear in the expression~\eqref{eq:XSEproduct}; here we have neglected this fact.  We leave such improvements  as an open problem for future work.

\section{Bosonic decoupling with arbitrary pulse times}\label{app:noisespectrum}
Here we consider the problem of decoupling a single bosonic mode from a bosonic bath with 
multiple modes. We will denote the system's mode operators by $\{Q,P\}$ and those of the bath by  $\{Q_j,P_j\}_j$. 
We assume that initially, system and environment are in a product state and that the environment is in the thermal state $\rho^E_\beta=e^{-\beta H_0^E}/\tr\left(e^{-\beta H_0^E}\right)$ at inverse temperature~$\beta$. The model  we consider is described by the original time independent Hamiltonian
\begin{align}\label{eq:Horigtemperature}
&\Horig= Q \sum_j \lambda_j Q_j +H_0^E\quad\textrm{where}\\
&\quad H_0^E=\tfrac{1}{2} \sum_j \omega_j (Q_j^2+P_j^2)
\end{align}
We analyze the resulting evolution
\begin{align}
\Ures(T)=&e^{i\Horig(1-\Delta_L)T}Ue^{i\Horig(\Delta_L-\Delta_{L-1})T}U\cdots\\& e^{i\Horig(\Delta_2-\Delta_1)T}Ue^{i\Horig\Delta_1T}\label{eq:resultingevol}
\end{align} after applying the Gaussian unitary 
\begin{align}
U=e^{i\frac{\pi}{2}(Q^2+P^2)}\label{eq:uplsunit}
\end{align} (acting on the mode operators as $UQU^*=-Q$ and $ UPU^*=-P$) at times $\Delta_jT$ for $j=1,\ldots,L$ and evolving under $\Horig$ at all other times.

We investigate how the efficiency of decoupling depends on the parameters $\lambda_j$ and $\omega_j$ and on the bath temperature $\beta$. We obtain identical expressions as in the analysis of $\pi$-pulse DD in the spin-boson model~\cite{uhrig2008exact,pasiniuhrigpra10}, see Theorem~\ref{thm:temperatureMres} below. Indeed, much of the following derivation closely mirrors the reasoning of~\cite{uhrig2008exact}, although the considered figure of merit is somewhat different: We directly compute the Gaussian CPTP map describing the system's evolution.
\begin{lemma}\label{lemma:temperatureUt}
	Define $U(t)=e^{-i\Horig t}Ue^{i\Horig t}$. Then
	\begin{align}
	U(t)&=U\exp(iK(t))=\exp(-iK(t))U\ .
	\end{align}
	where
	\begin{align}
	K(t)&=\frac{1}{\sqrt{2}}Q\sum_{j}\frac{\lambda_j}{\omega_j}\left(f(\omega_jt)a_j^\dagger+\overline{f(\omega_j t)}a_j\right)\ .
	\end{align} for $f(z)=ie^{-iz}-i$, where 
	\begin{align}
	a_j^\dagger=\frac{1}{\sqrt{2}}(Q_j-iP_j) \qquad a_j=\frac{1}{\sqrt{2}}(Q_j+iP_j)
	\end{align} 
	are the bosonic creation and annihilation operators of mode $j$.
\end{lemma}
\begin{proof}	
	It will again be convenient to describe this in terms of elements and generators of the symplectic group. 
	Let us assume that there are $n$ bath modes, and let us order the mode operators as
	$R=(Q,P,Q_1,\ldots,Q_n,P_1,\ldots,P_n)$. Then
	\begin{align}
	\Horig&=\frac{1}{2}\sum_{j,k}A_{j,k}R_jR_k
	\end{align}
	where $A$ is given by
	\begin{align}
	A&=\begin{pmatrix}
	0 & 0 & \lambda^T & 0_{1\times n}\\
	0 & 0 & 0_{1\times n} & 0_{1\times n}\\
	\lambda &  0_{n\times 1}&\Omega & 0_{n\times n}\\
	0_{n\times 1}& 0_{n\times 1} &0_{n\times n}&\Omega
	\end{pmatrix}
	\end{align}
	where  $\Omega=\mathsf{diag}(\omega_1,\ldots,\omega_n)$
	and $\lambda=\begin{pmatrix}\lambda_1 &\lambda_2 & \cdots & \lambda_n\end{pmatrix}^T$. 
	The symplectic group element~$\Sorig(t)$ associated with $e^{it\Horig}$
	can be computed to be
	\begin{align}
	\Sorig(t)=e^{tAJ}&=
	\begin{pmatrix}
	1& x(t)  &v(t)^T & w(t)^T\\
	0&  1& 0_{1\times n} & 0_{1\times n}\\
	0_{n\times 1} & w(t) &\cos(\Omega t) & -\sin(\Omega t)\\
	0_{n\times 1} & v(t) & \sin(\Omega t) & \cos(\Omega t)
	\end{pmatrix}\label{eq:tajexponentiated}
	\end{align} where
	\begin{align}
	v(t)&:=\Omega^{-1}(\cos(\Omega t)-I)\lambda\\
	w(t)&:=-\Omega^{-1}\sin(\Omega t)\lambda\\
	x(t)&:=t \lambda^T\Omega^{-1}\lambda-\lambda^T \Omega^{-2}\sin(\Omega t) \lambda\ .
	\end{align}
	Let us consider 
	\begin{align}
	U^*U(t)=U^*e^{-i\Horig t}Ue^{i\Horig t}&=e^{i\tilde{H} t}e^{i\Horig t}
	\end{align}
	where $\tilde{H}=- U^*\Horig U=\frac{1}{2}\sum_{j,k}\tilde{A}_{j,k}R_jR_k$.
	We note that this agrees with the Hamiltonian~$\Horig$ up to the replacements $\omega_j\rightarrow-\omega_j$.
	Since only $v(t)$ and $x(t)$ change sign under the substitution $\omega_j\rightarrow-\omega_j$ we conclude that $e^{t\tilde{A}J}$ is obtained from $e^{t{A}J}$ by substituting $v(t)\to -v(t)$ and $x(t)\to -x(t)$. Then we can compute the symplectic group element $e^{t\tilde{A}J}e^{tAJ}$ associated with $U^*U(t)$.
	On the other hand, consider an operator of the form
	\begin{equation}\label{eq:K}
	K=\frac{1}{\sqrt{2}}Q\sum_j (\kappa_ja^\dagger_j+\bar{\kappa}_ja_j)\ .
	\end{equation}
	This can equivalently be expressed as
	\begin{align}
	K&=Q\sum_j (\mathsf{Re}(\kappa_j)Q_j+\mathsf{Im}(\kappa_j)P_j)=\frac{1}{2}\sum_{j,k}B_{j,k}R_jR_k\ .
	\end{align}
	where 
	\begin{align}\label{eq:B}
	B&=\begin{pmatrix}
	0 & 0 & \mathsf{Re}(\kappa)^T & \mathsf{Im}(\kappa)^T\\
	0 & 0 & 0_{1\times n} & 0_{1\times n}\\
	\mathsf{Re}(\kappa) & 0 & 0_{n\times n}& 0_{n\times n}\\
	\mathsf{Im}(\kappa) & 0 & 0_{n\times n}& 0_{n\times n}
	\end{pmatrix}\ 
	\end{align}
	where 
	$\kappa=\begin{pmatrix}\kappa_1 &\kappa_2 & \cdots & \kappa_n\end{pmatrix}^T$. 
	Computing $e^{BJ}$ we see
	that this is equal $e^{t\tilde{A}J}e^{tAJ}$
	for the choice $\kappa=-2w(t)+2iv(t)$. That is, for
	\begin{align}
	K(t)&=2Q\sum_{j=1}^n \left(-w_j(t)Q_j+v_j(t)P_j\right)\ ,
	\end{align}
	we find that $U^* U(t)=e^{i\tilde{H} t}e^{i\Horig t}=e^{iK(t)}$.
	Inserting the relations $Q_j=\tfrac{1}{\sqrt{2}}(a_j^\dagger+a_j)$ and $P_j=\tfrac{i}{\sqrt{2}}(a_j^\dagger-a_j)$ implies the claim. 
\end{proof}

Let us use $U(t):=e^{-i\Horig t}Ue^{i\Horig t}$ to rewrite the resulting evolution in~\eqref{eq:resultingevol} as
\begin{align}
\Ures(T)&=   e^{i\Horig T}  U(\Delta_LT)\cdots U(\Delta_2T)U(\Delta_1T)\ .
\end{align}
Inserting $U(t)=U\exp(iK(t))$ from Lemma~\ref{lemma:temperatureUt} and using the commutation relation $U\exp(iK(t))=\exp(-iK(t))U$ we obtain
\begin{align}
\Ures(T)=&e^{i\Horig T}
e^{-iK(\Delta_LT)}e^{iK(\Delta_{L-1}T)}e^{-iK(\Delta_{L-2}T)}\cdots\\&\cdots e^{(-1)^{L-1}iK(\Delta_2T)} e^{(-1)^L iK(\Delta_1T)}U^L\ .
\end{align}
For $L$ even, this becomes
\begin{align}
\Ures(T)&=e^{i\Horig T}\prod_{m=1}^L e^{-i(-1)^m K(\Delta_mT)}
\end{align}
since $U^2=I$. 
Observe also that
\begin{align}
[iK(t_1),iK(t_2)]&=i\varphi(t_1,t_2)I
\end{align}
for some scalar $\varphi(t_1,t_2)$ and thus 
\begin{align}
[K(t_1),[K(t_1),K(t_2)]]&=0\qquad\textrm{ for all }t_1,t_2\in\mathbb{R}\ .
\end{align}
With the CBH formula $e^{A}e^{B}=e^{A+B}e^{[A,B]/2}$ if $[A,[A,B]]=[B,[A,B]]=0$
we conclude that
\begin{align}
\Ures(T)&=e^{i\Horig T}e^{i\Kres(T)}e^{i\phires(T)}\label{eq:uresproductexpression}
\end{align}
where
\begin{align}
\Kres(T)&=-\sum_{j=1}^L (-1)^j K(\Delta_j T)\\
\phires(T)&=-\frac{1}{2}\sum_{j=1}^{L-1}\sum_{\ell=1}^j(-1)^{j+\ell}\varphi(\Delta_{j+1}T,\Delta_\ell T)\ .
\end{align}
The exact form of $\phires$ will not be needed here but we compute $\Kres(T)$ to be
\begin{align}
\Kres(T)&=-\frac{1}{\sqrt{2}}Q\sum_{j}\frac{\lambda_j}{\omega_j}
\left( f_L(\omega_j T)a_j^\dagger+ \overline{f_L(\omega_j T)} a_j\right)
\end{align}
where
\begin{align}
f_L(z)&=\sum_{m=1}^{L}(-1)^m f(z\Delta_m)=2i\sum_{m=1}^{L}(-1)^m e^{-iz\Delta_m} 
\end{align}
since $\sum_{m=1}^{L}(-1)^m=0$ for $L$ even.

The main result of this section is a full description of the system's resulting evolution when an arbitrary pulse sequence
consisting of multiple applications of the unitary~$U$ (cf.~\eqref{eq:uplsunit}) is used. It is given by a Gaussian channel, i.e., a completely positive trace-preserving map, and is thus specified by its action on covariance matrices (see Eq.~\eqref{eq:covariancematrixdefmjkrho} below).

\begin{theorem}\label{thm:temperatureMres}
	Suppose  the  system and bath are initially in the product state $\rho \otimes\rho^E_\beta$, where the system's state $\rho$ has covariance matrix~$M$ and~$\rho^E_\beta=e^{-\beta H_0^E}/\tr\left(e^{-\beta H_0^E}\right)$ is the thermal state of the environment at inverse temperature~$\beta$.  
	Consider the state~
	\begin{align}
	\rho^{\mathsf{res}}=\tr_E\left(\Ures(T) (\rho\otimes\rho^E_\beta)\Ures(T)^*\right)
	\end{align} of the system at time $T$, i.e.,~after application of the pulse from Eq.~\eqref{eq:uplsunit} at times $\{T\Delta_j\}_{j=1}^L$ and uncontrolled evolution under~\eqref{eq:Horigtemperature} in between. Then $\rho^{\mathsf{res}}$ has covariance matrix
	\begin{align}\label{eq:Mres}
	\Mres =& \begin{pmatrix}
	1 & \xres\\
	0 & 1
	\end{pmatrix} M \begin{pmatrix}
	1 & \xres\\
	0 & 1
	\end{pmatrix}^T+
	\begin{pmatrix}
	\yres& 0\\
	0 & 0
	\end{pmatrix}
	\end{align}
	where 
	\begin{align}
	\xres&= \sum_j \frac{\lambda_j^2}{\omega_j^2}\Big[T {\omega_j}-\sin(\omega_jT) +\sin(\omega_jT)\mathsf{Re}(y_L(\omega_j T))\\
	& \qquad\qquad-(\cos(\omega_jT)-1)\mathsf{Im}(y_L(\omega_j T))\Big)\Big]\label{eq:xres}\\
	\yres &=\sum_j \frac{\lambda_j^2}{\omega_j^2}\coth \left(\frac{\beta\omega_j}{2}\right) \big|y_L(\omega_j T) \big|^2 \label{eq:yres}
	\end{align} 
	and where
	\begin{equation}\label{eq:yL}
	y_L(z)=2\sum_{m=1}^L (-1)^m e^{iz\Delta_m}+1-e^{iz}\ .
	\end{equation}
\end{theorem}

We note that $\yres$ depends on $\beta$ whereas $\xres$ does not and that the matrix 
$\begin{pmatrix}
1 & \xres\\
0& 1
\end{pmatrix}$ is symplectic. Hence, if $\yres=0$, the evolution $M\mapsto \Mres$  
of the system's covariance matrix is described by a Gaussian unitary and, in particular, is decoupled. On the other hand, any value $\yres>0$ indicates that the system's evolution is non-unitary, with $\yres$ quantifying the degree of decoherence introduced. 
\begin{proof}
	The state~$\rho^E_\beta$ has covariance
	$M_\beta=D_\beta\oplus D_\beta$, where 
	\begin{align}
	D_\beta &=\mathsf{diag}(\coth(\beta\omega_1/2),\ldots,\coth(\beta\omega_n/2))\ .
	\end{align}
	The output covariance matrix~$M^{\mathsf{res}}$ of the system is obtained by
	taking the principal~$2\times 2$ submatrix of $\Sres(T) M\oplus (D_\beta\oplus D_\beta)\ \Sres(T)^T$.

	To compute~$\Sres(T)$, we consider the expression~\eqref{eq:uresproductexpression} for the associated Gaussian unitary~$\Ures(T)$. We may neglect the phase $e^{i\varphi^{\mathsf{res}}(T)}$ as we are only interested in the evolution of the covariance matrix. 	We note that the Hamiltonian~$\Kres(T)$ is again of the form~\eqref{eq:K} for $\kappa_j=-\frac{\lambda_j}{\omega_j}f_L(\omega_jT)$. Let 
	us denote the symmetric matrix~\eqref{eq:B} associated with $\Kres(T)$ by $B^\textsf{res}$.
	Then the resulting symplectic operation~$\Sres(T)$ is 
	\begin{align}
	\Sres(T)&=e^{TAJ}e^{B^\textsf{res}J}=\begin{pmatrix}
	1 & \xres & a^T & b^T\\
	0 & 1 & 0 & 0\\
	0 & c & \cos(\Omega T) & -\sin(\Omega T)\\
	0 & d & \sin(\Omega T) & \cos(\Omega T)
	\end{pmatrix}
	\end{align}
	where 
	\begin{align}
	\xres&=x(T)-v(T)^T\mathsf{Re}(\kappa)-w(T)^T\mathsf{Im}(\kappa)\label{eq:xresdef}\\
	a&=\mathsf{Im}(\kappa)+v(T)\\
	b&=-\mathsf{Re}(\kappa)+w(T)\\
	c&=w(T)-\cos(\Omega T)\mathsf{Re}(\kappa)+\sin(\Omega T)\mathsf{Im}(\kappa)\\
	d&=v(T)-\sin(\Omega T)\mathsf{Re}(\kappa)-\cos(\Omega T)\mathsf{Im}(\kappa)\ .
	\end{align}
	Here we inserted the exponential expression~\eqref{eq:tajexponentiated} for $e^{TAJ}$ and an analogous expression for $e^{B^{\textsf{res}}J}$. 	
	
	We next consider the output covariance matrix $\Sres(T) M\oplus (D_\beta\oplus D_\beta)\ \Sres(T)^T$ of the system and environment. Its principal~$2\times 2$ submatrix is given by~\eqref{eq:Mres}, with $\xres$ from~\eqref{eq:xresdef} and
	\begin{align}
	\yres =&
	\sum_j \frac{\lambda_j^2}{\omega_j^2}\coth \left(\frac{\beta\omega_j}{2}\right)
	\Big[ \big(\mathsf{Re}(f_L(\omega_jT))-\sin(\omega_jT)\big)^2 \\
	&\qquad\qquad+\big(\mathsf{Im}(f_L(\omega_jT))+1-\cos(\omega_jT)\big)^2\Big]
	\end{align}
	With the definition~\eqref{eq:yL} of $y_L(z)$ we compute
	\begin{align}
	\mathsf{Re}(f_L&(z))-\sin(z)=\\
	&2\sum_{m=1}^L (-1)^m \sin(z \Delta_m ) - \sin(z)=\mathsf{Im}(y_L(z))\\
	\mathsf{Im}(f_L&(z))-(\cos(z)-1)=\\
	&2\sum_{m=1}^L (-1)^m \cos(z\Delta_m ) +1- \cos(z)=\mathsf{Re}(y_L(z))
	\end{align}
	In summary, we obtain the expression~\eqref{eq:Mres} with 	$\xres$ and $\yres$ as in~\eqref{eq:xres} and~\eqref{eq:yres}, respectively, as claimed. 
\end{proof}
The quantity $\yres$ introduced in Theorem~\ref{thm:temperatureMres} fully captures the error of our decoupling scheme.  Introducing the noise spectrum
\begin{equation}
S_\beta(\omega)=\sum_j \lambda_j^2 \delta(\omega -\omega_j) \coth \left(\frac{\beta\omega}{2}\right)
\end{equation} we can reexpress this quantity as 
\begin{equation}\label{eq:yresnoisespectrum}
\yres=\int_0^\infty \frac{S_\beta(\omega)}{\omega^2} \big|y_L(\omega T) \big|^2 d\omega \ .
\end{equation} 
The chosen decoupling pulse sequence, i.e., the pulse application times $t_j=T\Delta_j$, enter this expression only through the definition of~$y_L$. We note that the scalar in~\eqref{eq:yresnoisespectrum} is identical to the expression $\chi(T)$ which appears in the analysis of $\pi$-pulse DD of the spin-boson model for a single qubit~\cite{pasiniuhrigpra10,uhrig2010efficient}: It was shown that~$\chi(T)$ characterizes the efficiency of $\pi$-pulse DD and in particular its dependence on the bath temperature or the high-frequency cutoff in the noise spectrum~$S_\beta(\omega)$. For hard high-frequency cutoffs the Uhrig times are optimal. In~\cite{uhrig2010efficient}, Uhrig and Pasini numerically found that for a soft high-frequency cutoff the optimal DD times are not those of UDD but close to those of periodic DD. 
In summary, these qubit DD results on the high-frequency cutoffs in the noise spectrum translate immediately to the bosonic setting considered in Theorem~\ref{thm:temperatureMres}.

\section{Parametrization of the symplectic group and its Lie algebra: Proof of Lemma~\ref{lem:parametrizesymplecticgroupalgebra}}\label{app:proofsymplecticparametr}
Here we provide a proof for Lemma~\ref{lem:parametrizesymplecticgroupalgebra}, i.e., we show that the matrices~$S_\alpha$ satisfy the properties \eqref{it:symplecticalgebra}--\eqref{it:adjointaction}, and we show that the homogenization pulses from the substitution rule~\eqref{eq:replacements} are passive Gaussian unitaries associated to elements of~$\{S_\beta|\ \beta\in\Gamma\}$.

\begin{proof}
\eqref{it:symplecticalgebra}
Define $\delta(\alpha)=|\{j\in \{0,\ldots,m\}\ |\ a_j=(1,1)\}| $ 
and let~$\Gamma$ be the set of sequences
$\alpha=(a_0,\ldots,a_m)\in (\mathbb{Z}_2^2)^{m+1}$
such that $\delta(\alpha)+\delta_{a_0\in \{(0,1),(1,0)\}}$ is odd.

First consider the equation~$X^TJ_{2^m}+J_{2^m}X=0$ satisfied by any element~$X$ of~$\mathfrak{sp}(2\cdot 2^m)$. For simplicity let us omit the index~$2^m$ and write~$J$ instead of~$J_{2^m}$.
Using the fact that~$x^T=x$, $y^T=-y$ and~$z^T=z$, as well as the commutation/anticommutation relations
\begin{align}\label{eq:comrelxyz}
	[J,y]&=[J,I]=0\\
	\{J,x\}&=\{J,z\}=0\ ,
\end{align}
we find that
\begin{align}
	S_\alpha^TJ+J S_\alpha&=(-1)^{\delta(\alpha)} S_\alpha J+JS_\alpha\\
	&=\left((-1)^{\delta(\alpha)+\delta_{a_0\in \{(0,1),(1,0)\}}}+1\right)JS_\alpha 
\end{align}
for any sequence $\alpha=(a_0,\ldots,a_m)\in (\mathbb{Z}_2^2)^{m+1}$. Hence by definition of~$\Gamma$, the matrix~$S_\alpha$ is in~$\mathfrak{sp}(2\cdot 2^m)$ for every~$\alpha\in\Gamma$. Second, we note that all~$S_\alpha$ are linearly independent by definition. To show that they form a basis of~$\mathfrak{sp}(2\cdot 2^m)$ we compute the number of elements in~$\Gamma$ and compare it to the dimension $n_S(2n_S+1)=2\cdot2^{2m}+2^m$ of the Lie algebra~$\mathfrak{sp}(2\cdot 2^m)$. Any element~$\alpha=(a_0,a_1,\ldots,a_m)\in\Gamma$ satisfies 
	\begin{equation}
	\delta \big((a_1,\ldots,a_m)\big)\begin{cases}
	\textrm{ is even} & \textrm{if } a_0\in \{(0,1),(1,0),(1,1)\} \\
	\textrm{ is odd} & \textrm{if } a_0=(0,0)
	\end{cases} .\label{eq:deltaGamma}
	\end{equation} Let us therefore consider the number of vectors~$(a_1,\ldots,a_m)\in\mathbb{Z}_2^{2m}$ such that $\delta \big((a_1,\ldots,a_m)\big) $ is even or odd, which we denote by~$e_m$ or~$o_m$, respectively. These are recursively defined as~$e_{m+1}=e_1 e_m+ o_1o_m$ and $o_{m+1}=e_1o_m+o_1e_m$ where~$e_1=3$ and~$o_1=1$. An induction on~$m$ shows that $e_{m}=2^{{2m-1}}+2^{m-1}$.
	Using~Eq.~\eqref{eq:deltaGamma} the number of elements in~$\Gamma$ is given by $|\Gamma|=3e_m+o_m=e_{m+1}$. Hence it is equal to the dimension~$2\cdot2^{2m}+2^m$ of the symplectic algebra~$\mathfrak{sp}(2\cdot 2^m)$. In summary,~$\{S_\alpha\}_{\alpha\in \Gamma}$ is a basis of~$\mathfrak{sp}(2\cdot 2^m)$.
	
\eqref{it:symplecticgroup}
	By orthogonality of the matrices~$I,x,y,z$, for all~$\alpha\in (\mathbb{Z}_2^2)^{m+1}$ the matrices~$S_\alpha$ are also orthogonal. Any element~$S$ of the symplectic group has to satisfy the equation~$S^TJS=J$. Using the orthogonality of~$S_\alpha$ and the commmutation/anticommutation relations from Eq.~\eqref{eq:comrelxyz} we find that
	\begin{align}
		S_\alpha^T J S_\alpha&=(-1)^{\delta_{a_0\in \{(0,1),(1,0)\}}} S_\alpha^TS_\alpha J\\
		&=(-1)^{\delta_{a_0\in \{(0,1),(1,0)\}}}J \ .
	\end{align} This is equal to $J$ if and only if  $\delta_{a_0\in \{(0,1),(1,0)\}}=0$ or equivalently $a_0\in \{(0,0),(1,1)\}$. Hence the matrix~$\{S_\alpha\}$ is symplectic orthogonal for~$\alpha\in\tilde{\Gamma}$.
	
\eqref{it:symplecticformasssmatrix} With the definition of $y$, we write $y_0$ as
    \begin{equation}
        y_0=y\otimes  I_{2^{m}}=\begin{pmatrix}
            0& -1\\
            1& 0
        \end{pmatrix}\otimes I_{2^{m}}
    \end{equation} and the matrix $J_{2^m}$ defining the symplectic form as
    \begin{equation}
        J_{2^m}=\begin{pmatrix}
            0_{2^m}& I_{2^{m}}\\
            -I_{2^{m}}& 0_{2^{m}}
        \end{pmatrix}=-y \otimes  I_{2^{m}}
    \end{equation} where $0_{2^{m}}$ and $I_{2^{m}}$ denote the ${2^{m}}\times {2^{m}}$ zero and identity matrices, respectively. Since $y_0=S_{(1,1,0,\ldots,0)}$ the claim follows.

\eqref{it:adjointaction}
The commutation relations
\begin{align}\label{eq:commrelxyz}
\begin{matrix}
&x^Tyx=-y \qquad y^Txy=-x \qquad z^Txz=-x&\\
&x^Tzx=-z\qquad \,y^Tzy=-z \qquad  z^Tyz=-y&
\end{matrix}
\end{align} between the matrices~$x$, $y$, and~$z$
imply the desired relation~\eqref{eq:adjointactionsymplectic} for the adjoint action.
\end{proof}

Let us now consider the unitaries obtained by the substitution rule~\eqref{eq:replacements}. They are products of the unitaries~$U_{y_0}$, $U_{x_i}$, $U_{y_i}$, and~$U_{z_i}$ (where~$i=1,\ldots,m$) defined in~\eqref{eq:uy0def} and~\eqref{eq:uxjuzjdef}. The associated symplectic matrices are
\begin{align}
    y_0&=y\otimes I^{\otimes m}\\
    x_i&=I\otimes I^{\otimes i-1}\otimes x\otimes I^{\otimes m-i}\\
    y_i&=I\otimes I^{\otimes i-1}\otimes y\otimes I^{\otimes m-i}\label{eq:xjzjdefinition}\\
    z_i&=I\otimes I^{\otimes i-1}\otimes z\otimes I^{\otimes m-i}. 
\end{align} respectively. 
The matrix $x_i$ acts as a product of~$2^{m-1}$ SWAP operations between pairs of modes, $y_0$ (respectively $z_i$) as a tensor product of $2^m$ (respectively $2^{m-1}$) identical single-mode orthogonal symplectic operations. We remark that~$y_0$ and~$y_i$ (for $i=1,\ldots,m$) are different in nature: Whereas the former is a tensor product of single-mode operations, the latter being equal to $x_iz_i$ acts as a product of two-mode SWAP gates  and single-mode gates.
Products of these matrices can be written as~$S_\alpha$ where $\alpha=(a_0,\ldots,a_m)$ and~$a_0\in\{(0,0),(1,1)\}$, i.e.,~they are elements of~$\{S_\beta|\ \beta\in\tilde{\Gamma}\}$ and hence orthogonal symplectic. Because such matrices are associated with passive Gaussian unitaries, the homogenization pulses chosen according to the substitution rule~\eqref{eq:replacements} are passive Gaussian unitaries.

\section{Proof of homogenization}\label{app:proofhomogen}

In this section we provide a rigorous proof of the claim that any universal~$N$-th order DD scheme for~$m+1$ qubits that uses Pauli pulses induces an~$N$-th order bosonic homogenization scheme for~$2^m$ modes via the replacement rules~\eqref{eq:replacements}. More precisely, we
first show that the deduced bosonic pulse sequence falls into the framework of Section~\ref{app:generaldecoupling}. Second we show that the $N$-th order decoupling property~\eqref{eq:conditionsqubitdec} of the functions $F^\textsf{qubit}_\alpha$ translates to the homogenization property~\eqref{eq:conditionshom1} of the induced functions~$F^\mathsf{bos}_\alpha$. In last paragraph, we apply this result to the qubit NUDD scheme to obtain the bosonic homogenization sequence from Theorem~\eqref{thm:homogenization}.

Consider an $(m+1)$-qubit DD scheme with Pauli pulses~$U_j=\sigma_{\beta(j)}$ applied at times~$t_j$ for $j=1,\ldots,L$. On the level of symplectic group elements, the substitution rule~\eqref{eq:replacements} can be formulated as
\begin{equation}\label{substitutionrule_app}
    \sigma_\beta\mapsto S_{\beta'} \quad \textrm{for all } \beta=(b_0,b_1,\ldots,b_m)\in (\mathbb{Z}_2^2)^{m+1}
\end{equation}
 where the index $\beta'\in \tilde{\Gamma}$ is related to $\beta$ via
\begin{align}
    &\beta'=(b_0',b_1,\ldots,b_m) \  \textrm{ and }\label{eq:beta'beta}\\
    & b_0'=\begin{cases} 
            (1,1) & \textrm{ if } b_0\in \{(1,0),(1,1)\}\label{eq:beta'beta2}\\
            (0,0) & \textrm{ if } b_0\in \{(0,0),(0,1)\} 
        \end{cases}
\end{align}
Hence, applying~\eqref{eq:replacements} to the above $(m+1)$-qubit DD scheme results in a bosonic pulse sequence consisting in pulses~$S_{\beta'(j)}$ applied at times~$t_j$. We note that due to the replacement~$(\sigma_z)_0\mapsto S_{(0,\ldots,0)}=I_{2^m}$, the number of required homogenization pulses may be reduced compared to the original~$L$.

For completeness, let us briefly show that the deduced bosonic homogenization scheme falls into the framework from Section~\ref{app:generaldecoupling}:
Since the system is assumed to be decoupled from the environment, $\Xorig(t)$ is of direct-sum form~$\Xorig(t)=\Xorig_S(t)\oplus \Xorig_E(t)$ where the system part~$\Xorig_S(t)$ (the one to be homogenized) can be written as~\eqref{eq:systembathcoupling} for a one-dimensional environment and the operators $B_\alpha$ satisfy~\eqref{eq:analyticexpansionmodeops} by assumption. The adjoint action of symplectic group elements $S_\beta$ for $\beta\in\tilde{\Gamma}$ is specified by Eq.~\eqref{eq:adxinvidentity} as shown in the the previous section.
Furthermore, these group elements are infinitesimally generated up to signs as the matrices $y_0,x_i,y_i,z_i$ from Eqs.~\ref{eq:xjzjdefinition} can be expressed up to overall signs as exponentials of elements in~$\mathfrak{sp}(2\cdot 2^m)$: explicitly, we have 
\begin{align}
\exp\left(\frac{\pi}{2} y_0\right)=-y_0 \ ,\quad
&\exp\left(\frac{\pi}{2} y_0(x_i+I^{\otimes m+1})\right)=- x_i\ ,\\
\exp\left(\frac{\pi}{2} y_i\right)=- y_i \ ,\quad
&\exp\left(\frac{\pi}{2} y_0(z_i+I^{\otimes m+1})\right)=- z_i\ ,
\end{align}
for $i=1,\ldots,m$. The arguments of the exponential functions are indeed elements of $\mathfrak{sp}(2\cdot 2^m)$ as can be seen from the definitions of these matrices~\eqref{eq:xjzjdefinition} and property~\eqref{it:symplecticalgebra} of $\mathfrak{sp}(2\cdot 2^m)$.
Additionally the product of all pulses applied up to time $T$ amounts to the identity operator, again up to signs as this property is inherited from the corresponding $(m+1)$-qubit DD scheme: if in the multi-qubit setting~$\prod_{j=1}^L \sigma_{\beta(j)}=e^{i\varphi}I_{2^{m+1}}$ then the substitution rule~\eqref{substitutionrule_app} yields $\prod_{j=1}^L S_{\beta(j)}=\pm I_{2^{m+1}}$ in the bosonic setting which implies that $\Stf_S(T)=\pm \Sres_S(T)$. Thereby, all assumptions of the decoupling framework in Section~\ref{app:generaldecoupling} are satisfied. Hence, we can conclude that the toggling frame generator takes the form as derived in the main article in Eq.~\eqref{eq:togglinggenerator}.

The following lemma shows that if the original qubit DD scheme achieves $N$-th order decoupling, then the substitution rule~\eqref{eq:replacements} transform it into an $N$-th order homogenization scheme.
\begin{figure*}[t]
	\includegraphics[width=\textwidth]{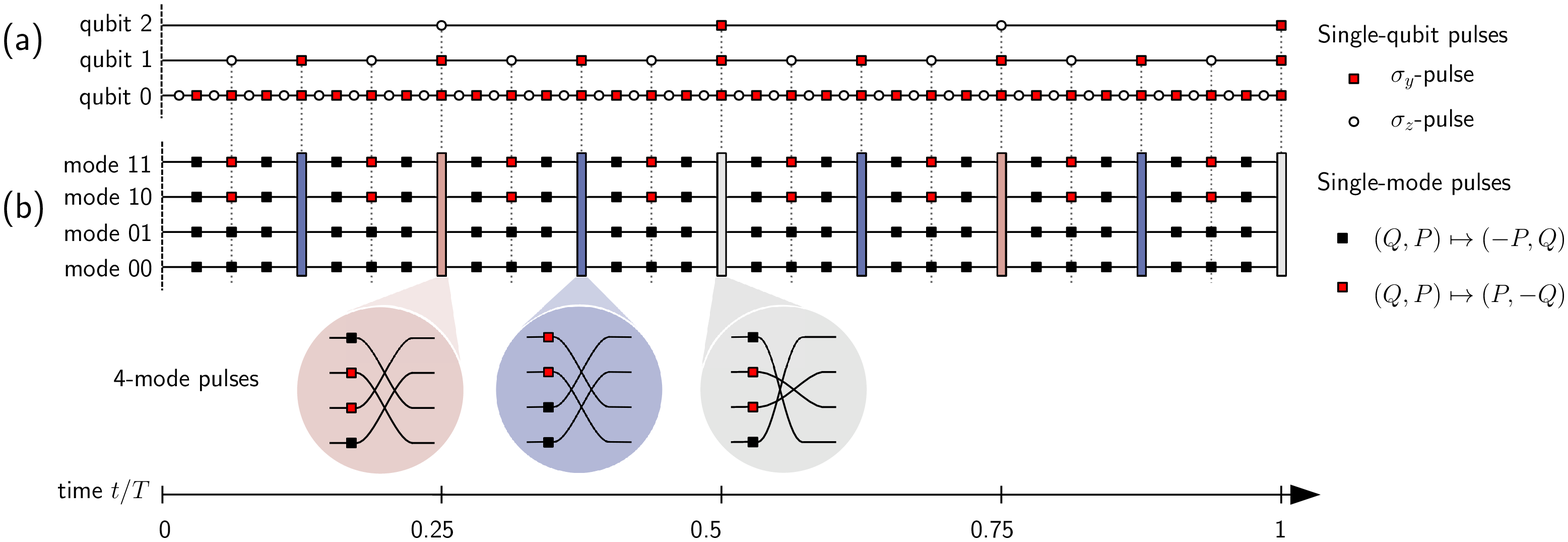}
	\caption{Bosonic homogenization scheme~(B) of suppression order~$N=1$ for four modes and the corresponding first order NUDD scheme~(A). The evolution under the decoherence Hamiltonian (horizontal straight lines) is interleaved with control pulses (circles and boxes); different four-mode gates (represented as red, blue and grey boxes) act as products of single-mode gates and SWAP gates between pairs of modes as represented by the insets below the figure.
		\label{fig:bosonic_hom_4modesorder1}}
\end{figure*}

\begin{lemma}\label{lem:boshomqubitdec}
    For a family of times $\{t_j\}_{j=1}^L$ and a function $\beta:\{0,\ldots,L\}\to (\mathbb{Z}_2^2)^{m+1}$ let~$\{F^\mathsf{qubit}_\alpha\}_{\alpha\in(\mathbb{Z}_2^2)^{m+1}}$ be defined as
    \begin{equation}\label{eq:qubitFalpha}
        F^\mathsf{qubit}_\alpha(t/T)=(-1)^{\sum_{j: t_j\le t} \langle \alpha, \beta(j)\rangle} \quad \textrm{for } t\in [0,T]\ .
    \end{equation} Assume there is a number~$N\in\mathbb{N}$ such that for all~$s\in \mathbb{N}$, $r_1,\ldots,r_s\in \mathbb{N}_0$, and $\alpha_1, \ldots, \alpha_s\in (\mathbb{Z}^2_2)^{m+1}$ these functions satisfy the decoupling condition~\eqref{eq:conditionsqubitdec}. 
    Define
    \begin{equation}\label{eq:homFalpha}
        F^\mathsf{bos}_\alpha(t/T)=(-1)^{\sum_{j: t_j\le t} \langle \alpha, \beta'(j)\rangle} \quad \textrm{for } t\in [0,T]\ .
    \end{equation} for all~$\alpha\in \Gamma$ where~$\beta'(j)$ is related to~$\beta(j)$ as specified in~\eqref{eq:beta'beta} and~\eqref{eq:beta'beta2}.
    Then the homogenization condition~\eqref{eq:conditionshom1} holds for~$F^\mathsf{bos}_\alpha$.
\end{lemma}
\begin{proof}
    Let~$\alpha\in\Gamma\subset (\mathbb{Z}_2^2)^{m+1}$. We use the notation~$a_0=(c,d)\in\mathbb{Z}^2_2$ for the zero component of~$\alpha$ in this proof. Observe that the functions~\eqref{eq:qubitFalpha} and~\eqref{eq:homFalpha}
    only differ in the expressions~$\langle \alpha,\beta'(j)\rangle$ and $\langle \alpha,\beta(j)\rangle$. Due to~\eqref{eq:beta'beta2} a summand of the former satisfies (modulo 2)
    \begin{align}\label{eq:a_0b_0'}
        a_0\begin{pmatrix}0&1\\-1&0\end{pmatrix} b_0'(j)^T= a'_0\begin{pmatrix}0&1\\-1&0\end{pmatrix} b_0(j)^T
    \end{align} where we defined $a'_0=(0,d\oplus c)$. Then the whole expression can be written as
    \begin{align}
        \langle \alpha, \beta'(j)\rangle& = \left[ \sum_{k=0}^m a_k \begin{pmatrix}0&1\\-1&0\end{pmatrix}b'_k(j)^T \right]\mod 2\\
        &=\langle \alpha', \beta(j)\rangle
    \end{align}
    where we defined $\alpha'=(a_0',a_1,\ldots,a_m)$. In conclusion, we find
    \begin{equation}
        F^\textsf{bos}_\alpha(t)=F^\textsf{qubit}_{\alpha'}(t)\quad \textrm{ for all } t\in[0,T]\ . 
    \end{equation} where $\alpha'=(c,c,0,\ldots,0)\oplus\alpha$.
    We note that the matrices~$S_\alpha$ satisfy~$\prod_{k=1}^sS_{\alpha_k}=\pm S_{\bigoplus_{k=1}^s \alpha_k}$ which can be easily seen from the commutation relations between~$x$, $y$, and~$z$. Additionally, the identity matrix is equal to~$S_{(0,\ldots,0)}$ and the matrix~$J_{2^m}$ defining the symplectic form is given by~$-S_{(1,1,0,\ldots,0)}$ (cf.~\eqref{it:symplecticformasssmatrix}). Therefore the condition
    \begin{align}
        &\prod_{k=1}^sS_{\alpha_k}\notin\{\pm I_{2\cdot 2^m},\pm J_{2^m}\}\ \textrm{ is equivalent to }\\
        &\bigoplus_{k=1}^s \alpha_k \notin\{ (1,1,0,\ldots,0),(0,\ldots,0)\}\ .
    \end{align}
    
    Now let~$s\in\mathbb{N}$, $r_1,\ldots,r_s\in\mathbb{N}_0$ and $\alpha_1,\ldots,\alpha_s\in \Gamma$ be such that $s+\sum_{k=1}^sr_k\le N$ and
    \begin{equation}\label{eq:sumalpha}
        \bigoplus_{k=1}^s \alpha_k \notin\{ (1,1,0,\ldots,0),(0,\ldots,0)\} \ .
    \end{equation} For $k=1,\ldots,s$ we introduce the notation $\alpha_k=(a_0(k),a_1(k)\ldots,a_m(k))$ and $a_0(k)=(c(k),d(k))\in \mathbb{Z}_2^2$. 
    Due to the definition of $\alpha'$ we have
    \begin{equation}\label{eq:sumalpha'}
        \bigoplus_{k=1}^s \alpha_k=\bigoplus_{k=1}^s \alpha'_k \oplus \big( \kappa, \kappa,0,\ldots,0\big)
    \end{equation} where we defined~$\kappa=\bigoplus_{k=1}^s c(k)$. In both cases~$\kappa=0$ and~$\kappa=1$,
    combining~\eqref{eq:sumalpha} and~\eqref{eq:sumalpha'} gives that $\bigoplus_{k=1}^s \alpha'_k\neq (0,\ldots,0)$. Then property~\eqref{eq:conditionsqubitdec} of the qubit DD functions~$F^\textsf{qubit}_{\alpha'}$ implies that 
    \begin{equation}
        \cF^{\vec{r}}_{\vec{\alpha'}} (\{F^\textsf{qubit}_{\alpha'}\})=0
    \end{equation}
    But due to the equality between~$F^\mathsf{qubit}_\alpha$ and~$F^\mathsf{bos}_\alpha$, then also $\cF^{\vec{r}}_{\vec{\alpha}} (\{F^\textsf{bos}_{\alpha}\})=0$. In conclusion the functions~$F^\mathsf{bos}_\alpha$ satisfy~\eqref{eq:conditionshom1}.
\end{proof}

What still remains to be proven is that a symplectic evolution of the form 
\begin{equation}\label{eq:sresdecoupledhomc1c2}
\Sres(T)=(c_1 I_{2n_S}+c_2 J_{n_S})\oplus S_E+O(T^{N+1})
\end{equation} for $c_1,c_2\in\mathbb{R}$ and $S_E\in \Sp(2n_E)$ describes a decoupled and homogenized evolution, i.e., that it can be written as $\Sres(T)=e^{\omega T J_{n_S}}\oplus S_E+O(T^{N+1})$ for some $\omega\in\mathbb{R}$. 
Indeed, the fact that~$\Sres(T)$ is symplectic and~\eqref{eq:sresdecoupledhomc1c2} imply that 
$	(c_1^2+c_2^2) J_{n_S}=J_{n_S}+O(T^{N+1})$ that is,  $c_1^2+c_2^2=1+\epsilon$ where $\epsilon=O(T^{N+1})$. 
Let $\omega\in\mathbb{R}$ be such that  $\cos \omega T=c_1/\sqrt{1+\epsilon}$ and $\sin\omega T=c_2/\sqrt{1+\epsilon}$. Then we have, using  $e^{\omega T J_{2^m}}=\cos(\omega T)I_{S}+\sin(\omega T)J_{2^m}$,
\begin{align}
\big\|(c_1I_{2n_S}+c_2 &J_{n_S})-e^{\omega T J_{n_S}}\big\| \\
& =\|(1-\frac{1}{\sqrt{1+\epsilon}}) (c_1 I_{2n_S}+c_2 J_{n_S})\|\\
&=O\left(|1-\frac{1}{\sqrt{1+\epsilon}}|\cdot (|c_1|+|c_2|)\right)\\ 
&\leq O\left(| 1-\frac{1}{\sqrt{1+\epsilon}}|\cdot \sqrt{ c_1^2+c_2^2}\right)\\
&= O\left(|1-\frac{1}{\sqrt{1+\epsilon}}|\cdot \sqrt{1+\epsilon}\right)=O(\epsilon)\ 
\end{align}
by the Cauchy-Schwarz-inequality. The claim then follows from the triangle inequality.

{\em Bosonic homogenization from the NUDD sequence.}
To derive Theorem~\ref{thm:homogenization}, it suffices to apply Lemma~\ref{lem:boshomqubitdec} to NUDD for $m+1$ qubits. We use the convention that qubit~$0$ is associated with the lowest concatenation level. This choice guarantees that the  substitution rule~\eqref{eq:replacements} achieves a maximal reduction of pulses, resulting in a $2^m$-mode bosonic homogenization scheme using $(N+1)^{2m+1}$~Gaussian unitaries.  An example of a bosonic homogenization scheme constructed from the NUDD scheme in this way is shown in Fig.~\ref{fig:bosonic_hom_4modesorder1}.

We remark that if a priori knowledge about the uncontrolled Hamiltonian is available, one can selectively  suppress stronger interactions. This can be realized, e.g., by using different suppression orders at various nesting levels in the recursive NUDD construction, see~\cite{wangliu11}.
The same reasoning can be extended to the bosonic setting. However, if decoupling is applied in conjunction with homogenization, prior information about the original uncontrolled Hamiltonian needs to be converted into prior information about the effective decoupled evolution. This appears to require a non-trivial analysis.

\section{Quadratic Hamiltonians with linear terms}\label{app:linearterms}
In this section, we discuss the effect of our schemes, both decoupling and homogenization, on system-environment interactions of the form
\begin{align}
\Horig(t)&=\frac{1}{2}\sum_{j,k=1}^{2(n_S+n_E)}A_{j,k}(t)R_j R_k\nonumber\\
&\qquad +\sum_{j=1}^{2(n_S+n_E)}b_{j}(t)R_j\label{eq:quadraticlinear}
\end{align} which are quadratic in the mode operators $R=(Q_1^S,\ldots,Q_{n_S}^S,P_1^S,\ldots,P_{n_S}^S,Q_1^E,\ldots,Q_{n_E}^E,P_1^E,\ldots,P_{n_E}^E)$ of system and environment,
but may include additional linear terms as parametrized by a time-dependent vector $b(t)\in\mathbb{R}^{2(n_S+n_E)}$ compared to Eq.~\eqref{eq:Hamiltonian}. We note that the evolution generated by~\eqref{eq:quadraticlinear}, together with our Gaussian control pulses, leads to a Gaussian unitary operation~$\Ures(T)$ which is characterized by its action on covariance matrices and displacement vectors of a given state~$\rho$. The latter are defined as 
\begin{align}
M_{j,k}(\rho)&=\tr(\{R_j-d_j,R_k-d_k\}\rho)\label{eq:covariancematrixdefmjkrho}\\
d_j(\rho)&=\tr(R_j\rho)
\end{align}
for $j,k\in \{1,\ldots,2(n_S+n_E)\}$.  The resulting unitary~$\Ures(T)$ acts as
\begin{align}
M\left(\rhores(T)\right)&=\Sres(T)M(\rho)\Sres(T)^T\\
d(\rhores(T))&=\Sres(T)d(\rho)+\zetares(T)\ ,
\end{align}
where $\rhores(T)=\Ures(T)\rho\Ures(T)^*$ is the time-evolved state. In this expression, $\Sres(T)\in \Sp(2(n_S+n_E))$ is symplectic, whereas $\zetares(T)\in\mathbb{R}^{2(n_S+n_E)}$ defines a displacement in phase space.

Here we argue that the constructed pulse sequence decouples respectively homogenizes the evolution of covariance matrices, i.e., the matrix~$\Sres(T)$ has the same structure as described in Theorem~\ref{thm:decoupling} respectively Theorem~\ref{thm:homogenization}.  In particular, the presence of linear terms in the Hamiltonian~\eqref{eq:quadraticlinear} has no effect on the evolution of second moments, and the corresponding evolution is decoupled from the environment respectively given by decoupled and homogenized oscillators up to the considered suppression order.  Indeed, this follows from the simple observation that 
the time evolution of the covariance matrix~$M(t):=M(\rho(t))$ for a state~$\rho(t)$ evolving under a time-dependent Hamiltonian as in~\eqref{eq:quadraticlinear} is governed by the differential equation
\begin{align}
\dot{M}(t)&= 2(-JA(t)M(t)+M(t)A(t)J)\ .
\end{align}
Importantly, this equation has no dependence on the vector~$b(t)$ associated with the linear terms. Furthermore, application of  unitary Gaussian  decoupling pulses does not change the form~\eqref{eq:quadraticlinear} of 
the time-dependent generator. This implies the claim.

\end{document}